%% file: main.tex
\newcount\Comments  
\def\arxiv{1} 

\ifnum\arxiv=1  
\documentclass[12pt]{article}
\usepackage{amsmath,amsfonts,amsthm,amssymb,bm,bbm} 
\usepackage{fullpage,microtype}
\usepackage{graphicx,subfig}
\usepackage[square,numbers,sort&compress]{natbib} 
\usepackage{authblk}
\usepackage[ruled]{algorithm2e}

\SetArgSty{textrm}  

\newtheorem{theorem}{Theorem}
\newtheorem{lemma}{Lemma}
\newtheorem{mainresult}{Main Result}
\newtheorem{observation}{Observation}
\newenvironment{proofidea}{\par{\noindent \textit{Proof Idea:} }}{\qed\par}
\theoremstyle{definition}
\newtheorem{definition}{Definition}

\numberwithin{theorem}{section}
\numberwithin{lemma}{section}
\numberwithin{definition}{section}

\newcommand{\Paragraph}[1]{\paragraph{#1.}}

\else 
\documentclass[prodmode,license]{acmsmall-ec15}

\doi{http://dx.doi.org/10.1145/2764468.2764519}

\usepackage{bm}
\usepackage[ruled]{algorithm2e}

\usepackage{amsmath,amsfonts,amssymb,bbm} 
\usepackage{boxedminipage}
\usepackage{subfig}
\usepackage[square,numbers,sort&compress]{natbib} 
\SetArgSty{textrm}  

\newtheorem{mainresult}{Main Result}

\newtheorem{observation}{Observation}

\newenvironment{proofidea}{\par{\noindent \textit{Proof Idea:} }}{\qed\par}
\newcommand{\Paragraph}[1]{\paragraph{#1}}

\fi 

\usepackage{color}
\definecolor{darkgreen}{rgb}{0,0.5,0}
\definecolor{darkorange}{rgb}{0.7,0.3,0}
\newcommand{\ignore}[1]{}
\newcommand{\kibitz}[2]{\ifnum\Comments=1\textcolor{#1}{#2}\fi}

\newcommand{\cj}[1]{\kibitz{blue}{\noindent[CJ: #1]}}

\newcommand{\bo}[1]{\kibitz{darkorange}{\noindent[Bo: #1]}}

\newcommand{\Omit}[1]{}

\def\Pr{\mathop{\textnormal{Pr}}}

\def\X{\mathcal{X}}
\def\Y{\mathcal{Y}}
\def\Z{\mathcal{Z}}
\def\D{\mathcal{D}}
\def\H{\mathcal{H}}

\def\F{\mathcal{F}}

\def\A{\mathcal{A}}
\def\1{\mathbbm{1}}
\DeclareMathOperator*{\E}{\mathbb{E}}
\def\R{\mathbbm{R}}

\def\algcost{\gamma_{T,\A}}
\def\algcostopt{\algcost^*}
\def\algcostmax{\algcost^{\max}}

\def\err{\mathcal{L}}

\newcommand{\It}{\1[z_t \text{\textnormal{ recvd}}]}

\def\loss{\ell}

\begin{document}

\title{Low-Cost Learning via Active Data Procurement}

\ifnum\arxiv=1  

\author[*]{Jacob Abernethy}
\author[**]{Yiling Chen}
\author[***]{Chien-Ju Ho}
\author[**]{Bo Waggoner}
\affil[*]{University of Michigan}
\affil[**]{Harvard SEAS}
\affil[***]{UCLA}
\date{May 2015}

\else 
\author{
JACOB ABERNETHY\affil{University of Michigan}
YILING CHEN\affil{Harvard University}
CHIEN-JU HO\affil{UCLA}
BO WAGGONER\affil{Harvard University}
}
\fi 

\ifnum\arxiv=1  
\maketitle
\fi

\begin{abstract}
We design mechanisms for online procurement of data held by strategic agents for machine learning tasks.
We study a model in which agents cannot fabricate data, but may lie about their cost of furnishing their data.
The challenge is to use past data to actively price future data in order to obtain learning guarantees, even when agents' costs can depend arbitrarily on the data itself.
We show how to convert a large class of no-regret algorithms into online posted-price and learning mechanisms.
Our results parallel classic sample complexity guarantees, but with the key resource constraint being money rather than quantity of data available. With a budget constraint $B$, we give robust risk (predictive error) bounds on the order of $1/\sqrt{B}$.
In many cases our guarantees are significantly better due to an active-learning approach that leverages correlations between costs and data.

Our algorithms and analysis go through a model of no-regret learning with $T$ arriving pairs (cost, data) and a budget constraint of $B$, coupled with the ``online to batch conversion''.
Our regret bounds for this model are on the order of $T/\sqrt{B}$ and we give lower bounds on the same order.
\end{abstract}

\ifnum\arxiv=0  
\begin{bottomstuff}
Addresses: jabernet@umich.edu, yiling@seas.harvard.edu, cjho@ucla.edu, bwaggoner@fas.harvard.edu.
\end{bottomstuff}
\maketitle
\fi

\section{Introduction}
\label{sec:intro}
\input{intro}
\subsection*{Related Work}

\input{related}

\section{Statistical Learning with Purchased Data}
\label{sec:stat-learning}
\input{stat-learning}

\section{Tools for Converting Regret-Minimizing Algorithms}
\label{sec:noregret}
\input{noregret}

\section{Regret Minimization with Purchased Data}
\label{sec:noregret-strategic}
\input{noregret-strategic}

\section{Results for Statistical Learning}
\label{sec:tieback}
\input{tieback}

\section{Deriving Pricing and the ``at-cost'' Variant}
\label{sec:noregret-nonstrategic}
\input{noregret-nonstrategic}

\ifnum\arxiv=1  
\section{Examples and Experiments}
\label{sec:examples}
\input{sim}

\fi  

\section{Discussion and Conclusion}
\label{sec:conclusion}
\input{conclusion}

\section*{Acknowledgments}
The authors thank Mike Ruberry for discussion and formulation of the problem. Thanks to the organizers and participants of the 2014
Indo-US Lectures Week in Machine Learning, Game Theory and Optimization, Bangalore.

We thank the support of the National Science Foundation under awards CCF-1301976 and IIS-1421391.  Any opinions, findings, conclusions, or recommendations expressed here are those of the authors alone.

\newpage

{
\bibliographystyle{plainnat}
\bibliography{importancepricing}
}

\ifnum\arxiv=1  
\newpage
\appendix
\section*{Appendix}
\input{appendix}

\fi  

\end{document}

%% file: intro.tex

The rising interest in the field of Machine Learning (ML) has been strongly driven by the potential to generate economic value.
Firms seeking revenue optimizations can gather abundant data at low cost, apply a set of inexpensive algorithmic tools, and produce high-accuracy predictors that can massively improve future decision making.
The extent of the potential value that can be created by leveraging data for prediction is apparent in the multi-million dollar competition bounties offered by companies like Netflix and the Heritage Health Foundation, but perhaps even more so in the aggressive hiring of many ML experts by companies like Google and Facebook.

Much of the theoretical results in ML aim to measure, at least implicitly, the \emph{economic efficiency} of learning problems.
For example, in certain settings we have a reasonably thorough understanding of \emph{sample complexity} \cite{anthony2009neural} which gives us the precise tradeoff between $n$, the quantity of data at our disposal, and the error or loss rate we want to achieve.
Reducing error is always beneficial, of course, but must be weighed against the marginal cost of increasing $n$.

The measures of efficiency in ML have broadened in recent years, in particular because gathering data is typically orders of magnitude cheaper than labeling it.
This has led to the emergence of the \emph{active learning} paradigm \cite{balcan2010true,IWAL09,hanneke2009theoretical,settles2011theories,BBL06}.
Here, we imagine an interface between the learner and the label provider, where the learner may make label queries on data points in an online fashion.
By sequentially choosing which data to label, the learner can greatly reduce the number of labels required to learn~\cite{hanneke2009theoretical}.

A problem that has received little attention in the learning theory literature is the \emph{monetary efficiency} of learning when data have differing costs.
Indeed, real-world prediction tasks often require obtaining examples held by self-interested, strategic agents; these agents must be incentivized to provide the data they hold, and they have heterogeneous costs for doing so.

In this vein, the present paper seeks to address the following question:
\begin{quotation}
 In a world where data is held by self-interested agents with heterogeneous costs for providing it, and in particular when these costs may be arbitrarily correlated with the underlying data, how can we design mechanisms that are incentive-compatible, have robust learning guarantees, and optimize the cost-efficiency tradeoffs inherent in the learning problem?
\end{quotation}

This question is relevant to many real-world scenarios involving financial and strategic considerations in data procurement.
Here are two examples:
\begin{enumerate}
	\item In the development of a certain drug, a pharmaceutical company wishes to train a disease classifier based on data obtained by hospitals and stored in patients' medical records.
	These data are not public, yet the company can offer hospital patients financial incentives to contribute their private records.
  We note the potential for cost heterogeneity: the compensation required by patients may be correlated with the content of their medical data (\emph{e.g.} if they have the disease).
	\item Online retailers generally hope to know more about website visitors in order to better target products to customers.
	A retailer can offer to buy customers' demographic and social data, say in the form of access to their Facebook profile.
	But again, customers' willingness to sell may covary with their demographics data in an unknown way.
\end{enumerate}


\subsection*{From sample complexity to budget efficiency}
The classical problem in statistical learning theory is the following. We are given $n$ datapoints (examples) $z_1, \ldots, z_n \in \Z$ sampled from some distribution $\D$. Our goal is to select a hypothesis $h \in \H$ which ``performs well'' on unseen data from $\D$. We can specify performance in terms of a \emph{loss function} $\loss(h,z)$, and we write $\err(h)$, known as the \emph{risk} of $h$, to be the expectation of $\loss(h,z)$ on a random draw $z$ from $\D$. The goal is to produce a hypothesis $\bar{h}$ whose risk is not much more than that of $h^*$, the optimal member of $\H$. For example, in \emph{binary classification}, each data point consists of a pair $z = (x,y)$ where $x$ encodes some ``features'' and $y \in \{-1,1\}$ is the label; a hypothesis $h$ is a function that predicts a label for a given set of features; and a typical loss function, the ``0-1 loss'', is defined so that $\loss(h,(x,y)) = 0$ when $h(x) = y$ and $\loss(h,(x,y)) = 1$ otherwise.

Research in statistical learning theory attempts to characterize how well such tasks can be performed in terms of the \emph{resources} available and the inherent \emph{difficulty} of the problem. The resource is usually the quantity of data $n$. In binary classification, for instance, the difficulty or richness of the problem is captured by the ``VC-dimension'' $d$, and a famous result \citep{vapnik2000nature} is that there is an algorithm achieving the bound
\begin{equation}
	\err(\bar h) \leq \err(h^*) + O\left( \sqrt{\frac{d \log n}{n}} \right),  \label{eqn:traditional-risk}
\end{equation}
with very high probability over the sample $z_1, \ldots, z_n$.

In the present work we consider an alternative scenario: the learner has a fixed budget $B$ and can use this budget to purchase examples.
More precisely, on round $t$ of a sequence of $T$ rounds, agent $t$ arrives with data point $z_t$, sampled i.i.d. from some $\D$, and a cost $c_t \in [0,1]$. 
This cost $c_t$ is known only to the agent and can depend arbitrarily on $z_t$.
The learning mechanism may offer a (possibly randomized) menu of \emph{take-it-or-leave-it} prices $\pi_t : \Z \to \R_+$, with a possibly different price $\pi_t(z)$ for each data point $z$. The arriving agent observes the price $\pi_t(z_t)$ offered for her data and accepts as long as $\pi_t(z_t) \geq c_t$, in which case the mechanism pays the agent $\pi_t(z_t)$ and learns $(c_t, z_t)$.\footnote{We will discuss the interaction model further in Sections \ref{sec:stat-learning} and \ref{sec:conclusion}.} 
Our goal is to actively select prices to offer for different datapoints, subject to a budget $B$, in order to minimize the risk of our final output $\bar{h}$.

At a high level, our main result parallels the classical statistical learning guarantee in \eqref{eqn:traditional-risk}, but where the limited resource is the budget $B$ instead of the sample size $n$.
\begin{mainresult}[Informal] \label{mainresult}
For a large class of problems, there is an active data purchasing algorithm $\A$ that spends at most $B$ in expectation and outputs a hypothesis $\bar{h}$ satisfying,
 \[ \E \err(\bar{h}) ~ \leq ~ \err(h^*) + O\left(\sqrt{\frac{\algcost}{B}}\right) ,  \]
where $\algcost \in [0,1]$ is an algorithm-dependent parameter of the (cost, data) sequence capturing the \emph{monetary difficulty of learning} and the expectation is over the algorithm's internal randomness. 
\end{mainresult}
This bound depends on the quantity $\algcost$ which captures the \emph{monetary difficulty} of the problem at hand.
(We also need as prior knowledge a rough estimate of $\algcost$.)
This is in rough analogy with VC-dimension in classical bounds such as Equation \ref{eqn:traditional-risk}.
Similarly, the key resource constraint is now the budget $B$ rather than the quantity of data $n$.

It is important to note that $\algcost$ depends on the choice of algorithm $\A$.
However, our results also include simpler, algorithm-independent bounds.
For instance, replace $\algcost$ by $\sqrt{\mu}$, where $\mu$ is the mean of the arriving costs, and Main Result \ref{mainresult} continues to hold (and the only prior knowledge required is a rough estimate of $\mu$).
But $\algcost$ can be significantly smaller than $\sqrt{\mu}$ when there are particular correlations between the costs and the examples; indeed, we can have $\algcost \to 0$ even as $\mu$ stays constant.
This indicates a case in which the average cost of data is high, but due to beneficial correlations between costs and data, our mechanism can obtain all the data it needs for good learning very cheaply.
We give a thorough discussion of $\algcost$ in Section \ref{sec:gamma_interpretation}.

\subsection*{Overview of Techniques}


\begin{figure}
\includegraphics[width=\textwidth]{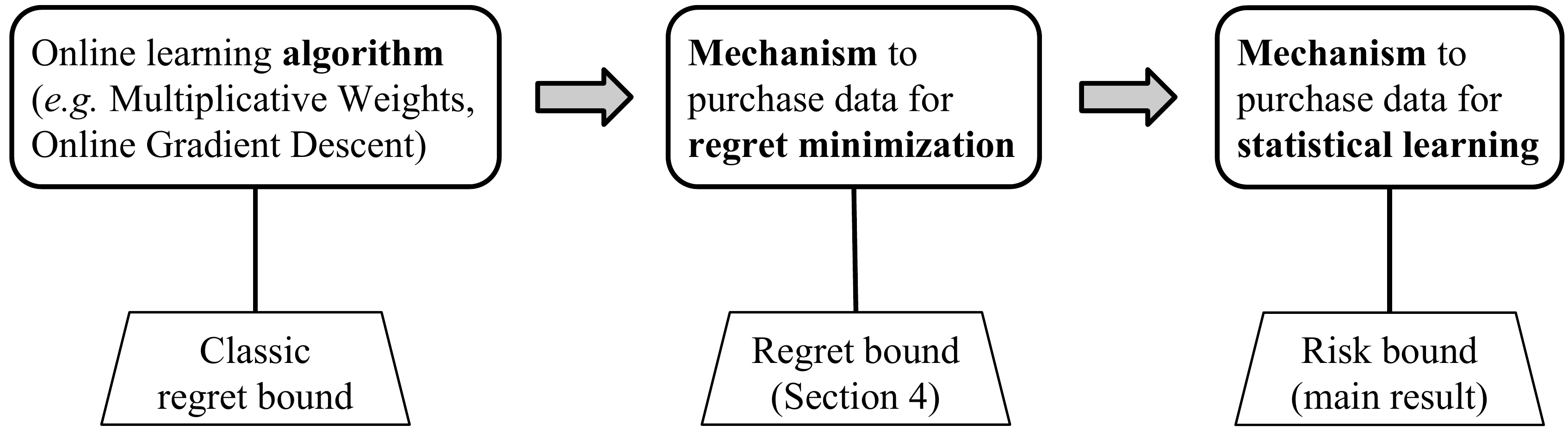}
\caption{\textbf{Algorithmic and analytic approach.}
  \small
  First, we convert Follow-the-Regularized-Leader online no-regret algorithms into mechanisms that purchase data for a regret-minimization setting that we introduce for purposes of analysis.
  Then, we convert these into mechanisms to solve our main problem, statistical learning. The mechanisms interact with the online learning algorithms as black boxes, but the analysis relies on ``opening the box''.}
\label{figure:conversions}
\end{figure}

Our general idea for attacking this problem is to utilize online learning algorithms (OLAs) for regret minimization~\cite{cesa2006prediction}.
These algorithms output a hypothesis or prediction at each step $t=1,\dots,T$, and their performance is measured by the summed loss of these predictions over all the steps.
The idea is that the hypotheses produced by the OLA at each step can be used both to determine the value of data during the procurement process and to generate a final prediction.

In Section \ref{sec:noregret}, we lay out the tools we need for a pricing and learning mechanism to interact with OLAs.
The first high-level problem is that, because of the budget constraint, our OLA will only see a small subset of the data sequence. We use the tool of \emph{importance-weighting} to give good regret-minimization guarantees even when we do not see the entire data sequence.
The second problem is how to aggregate the hypotheses of the OLA and convert its regret guarantee into a risk guarantee for our statistical learning setting. This is achieved with the standard ``online-to-batch'' conversion~\citep{cesa2004generalization}.

Given the tools of Section \ref{sec:noregret}, the key remaining challenge is to develop a pricing and learning strategy that achieves low regret.
We address this question in Section \ref{sec:noregret-strategic}.
We formally define a model of online learning for regret minimization with purchased data, in which the mechanism must output a hypothesis at each time step and perform well in hindsight against the entire data sequence, but only has enough budget to purchase and observe a fraction of the arriving data.
We defer until later our detailed analysis of this setting, derivation of a pricing strategy, and lower bounds.
At this point, we present our pricing strategy and regret guarantees for this setting.

In Section \ref{sec:tieback}, we give our main results: risk guarantees for a learner with budget $B$ and access to $T$ arriving agents.
These bounds follow directly by using the tools in Section \ref{sec:noregret} and regret-minimization results in Section \ref{sec:noregret-strategic}.

In Section \ref{sec:noregret-nonstrategic}, we develop a deeper understanding of the regret minimization setting.
We derive our pricing strategy from an in-depth analysis of a more analytically tractable variant of the problem, the ``at-cost'' setting, where the mechanism is only required to pay the cost of the arriving data point rather than the price posted.
For this setting, we are able to derive the optimal pricing strategy for minimizing the regret bound of our class of learning algorithms subject to an expected budget constraint.

We also complement our upper bounds by proving lower bounds for data-purchasing regret minimization.
These show that our mechanisms for the easier at-cost setting have an order-optimal regret guarantee of $\frac{T}{\sqrt{B}}\algcost$.
There is a small gap to our mechanisms for the main regret minimization setting, in which our guarantee is on the order of $\frac{T}{\sqrt{B}}\sqrt{\algcost}$ (recall that $\algcost \in [0,1]$, so this is a weaker guarantee).
The dependence $T/\sqrt{B}$ approaches the classic $\sqrt{T}$ regret bound when $B$ is large (approaching $T$).
When $B$ is small but still superconstant, we observe the perhaps counterintuitive fact that we can achieve $o(1)$ average regret per arrival while only observing an $o(1)$ fraction of the arriving data; in other words, we have ``no data, no regret.''

\ifnum\arxiv=0  
All proofs appear in the full version of the paper.
\fi  

%% file: related.tex
For ``batch'' settings in which all agents are offered a price simultaneously, pricing schemes for obtaining data have appeared in recent work, especially \citet{RothSchoenebeck12}, which considered the design of mechanisms for efficient estimation of a statistic.
However, this work and others in related settings \citep{LR12, GR11, cummings2015accuracy} consider offline solutions, \emph{e.g.} drawing a posted price independently for all data points.
We focus on an \emph{active} approach in which the marginal value of individual examples is estimated according to the current learning progress and budget.
A data-dependent approach to pricing data does appear in \citet{HIM14}, but that paper focuses on a quite different learning setting, a model of regression with noisy samples with a budget-feasible mechanism design approach.

Another difference from the above papers is that we prove risk and regret bounds rather than trying to minimize \emph{e.g.} a variance bound, and we also consider a broader class of learning problems.

\paragraph{Other related work.} Other works such as \citet{ghosh2014buying, dekel2008incentive, meir2012algorithms} focus on a setting in which agents may misreport their data (also see the peer-prediction literature).
We suppose that agents may misreport their costs but not their data.

Many of the ideas in the present work draw from recent advances in using importance weighting for the active learning problem \citep{IWAL09}. There is a wealth of theoretical research into active learning, including \citet{BHLZ10,BBL06,hanneke2009theoretical} and many others.

``Budgeted Learning'' is a somewhat related area of machine learning, but there the budget is not \emph{monetary}.
The idea is that we do not see all of the features of the data points in our set, but rather have a ``budget'' of the number of features we may observe (for instance, we may choose any two of the three features height, weight, age).

%% file: stat-learning.tex
In this section, we formally define the problem setting.
The body of the paper will then consist of a series of steps for deriving mechanisms for this setting with provable guarantees, which will finally appear in Section \ref{sec:tieback}.

We consider a statistical learning problem described as follows.
Our data points are objects $z \in \Z$.
We are given a hypothesis class $\H$ which we will assume is parameterized by vectors $\R^d$ but more broadly can be any Hilbert space endowed with a norm $\| \cdot \|$; for convenience we will treat elements $h \in \H$ as vectors which can be added, scaled, etc.
We are also given a loss function $\loss : \H \times \Z \to \R$ that is convex in $\H$.
We assume throughout the paper that the loss function is \emph{1-Lipschitz} in $h$; that is, for any $z \in \Z$ and any $h,h' \in \H$ we have $|\loss(h,z) - \loss(h',z)| \leq \| h - h' \|$.

In many common scenarios, $\Z$ is the space of pairs $(x,y)$ from the cross product $\X \times \Y$, with $x$ the feature input and $y$ the label, though in our setting $\Z$ can be a more generic object.
For example, in the canonical problem of \emph{linear regression}, we have that $\Z = \X \times \Y = \R^d \times \R$, the hypothesis class is vectors $\H = \R^d$, and the loss function is defined according to squared error $\loss(h,(x,y)) := (h^\top x - y)^2$.

The \textbf{data-purchasing statistical learning problem} is parameterized by the data space $\Z$, hypothesis space $\H$, loss function $\ell$, number of arriving data points $T$, and expected budget constraint $B$.
A \emph{problem instance} consists of a distribution $\D$ on the set $\Z$ and a sequence of pairs $(c_1,z_1), \dots, (c_T, z_T)$ where each $z_t$ is a data point drawn i.i.d. according to $\D$ and each $c_t \in [0,1]$ is the private cost associated with that data point.
The costs may be arbitrarily chosen, \emph{i.e.} we consider a worst-case model of costs.
(For instance, if costs and data are drawn together from a joint, correlated distribution, then this is a special case of our setting.)

In this problem, the task is to design a \emph{mechanism} implementing the operations ``post'', ``receive'', and ``predict'' and interacting with the problem instance as follows.
\begin{itemize}
  \item For each time step $t=1,\dots,T$:
  \begin{enumerate}
    \item The mechanism \emph{posts} a pricing function $\pi_t: \Z \to \R$, where $\pi_t(z)$ is the price posted for data point $z$.
    \item Agent $t$ arrives, possessing $(c_t,z_t)$.
    \item If the posted price $\pi_t(z_t) \geq c_t$, then agent $t$ accepts the transaction: The mechanism pays $\pi_t(z_t)$ to the agent and \emph{receives} $(c_t, z_t)$. If $\pi_t(z_t) < c_t$, agent $t$ rejects the transaction and the mechanism \emph{receives} a null signal.
  \end{enumerate}
  \item The mechanism outputs a \emph{prediction} $\bar{h} \in \H$.
\end{itemize}
Note that the mechanism is given the parameters $\Z$, $\H$, $\ell$, $T$, and $B$, but the problem instance is completely unknown to the mechanism prior to to the arrivals.
The design problem of the mechanism is how to choose the pricing function $\pi_t$ to \emph{post} at each time, how to update based on \emph{receiving} data, and how to choose the final \emph{prediction}.
The \emph{risk} or predictive error of a hypothesis is
 \[ \err(h) = \E_{z \sim \D} \loss(h,z)  \]
and the goal of the mechanism is to minimize the risk $\err(\bar{h})$ of its final hypothesis $\bar{h}$.
The benchmark is the optimal hypothesis in the class, $h^* = \arg\min_{h\in\H} \err(h)$.

The mechanism must guarantee that, for every input sequence $(c_1,z_1),\dots,(c_T,z_T)$, it spends at most $B$ in expectation over its own internal randomness.



\Paragraph{Agent-mechanism interaction}
The model of agent arrival and posted prices contains several assumptions.
First, agents cannot fabricate data; they can only report data they actually have to the mechanism.
Second, agents are rational in that they accept a posted price when it is higher than their cost and reject otherwise.
Third, we have an implementation of the mechanism that can obtain the agent's cost $c_t$ when the transaction occurs.

We emphasize that the purpose of this paper is not focused on the implementation of such a setting, but instead on developing active learning and pricing techniques and guarantees.
This is also intended as a simple and clean model in which to begin developing such techniques.
However, we briefly note some possible implementations.

In the most straightforward one, the mechanism posts prices directly to the agent who responds directly.
This would be a weakly truthful implementation, as agents have no incentive to misreport costs after they choose to accept the transaction.

One strictly truthful implementation uses a \emph{trusted third party} (TTP) that can facilitate the transactions (and guarantee the validity of the data if necessary).
For example, we could imagine attempting to learn to classify a disease, and we could rely on a hospital to act as the broker allowing us to negotiate with patients for their data.
Then the TTP/agent interaction could proceed as follows:
\begin{enumerate}
	\item Learning mechanism submits the pricing function $\pi_t$ to the TTP;
	\item Agent provides his data point $z_t$ and cost $c_t$ to the TTP;
	\item TTP determines whether $\pi_t(z_t) \geq c_t$ and, if so, instructs the learner to pay $\pi_t(z_t)$ to the agent and then provides the pair $(z_t,c_t)$ to the learner.
\end{enumerate}
Other possibilities for strictly truthful implementation include using a bit of cryptography (see Section \ref{sec:conclusion}).


%% file: noregret.tex
In this section we begin with the classic regret-minimization problem and a broad class of algorithms for this problem.
We then show how to apply techniques that convert these algorithms into a form that will be useful for solving the statistical learning problem with purchased data.
The only missing ingredient will then be a price-posting strategy, which will be presented in Section \ref{sec:noregret-strategic}.

\subsection{Recap of Classic Regret-Minimization}
In the classic regret-minimization problem, we have a hypothesis class $\H$ with the same assumptions as stated in Section \ref{sec:stat-learning}.
At each time $t=1,\dots,T$ the algorithm posts a hypothesis $h_t \in \H$.
Nature (the adversary, the environment, etc.) selects a $1$-Lipschitz convex loss function $f_t: \H \to \R$.\footnote{This definition of ``loss function'' is a departure from our main setting which involved $\loss(\cdot,\cdot)$. But we will use this somewhat more general setup by choosing $f_t(h) \propto \loss(h,z_t)$ for the datapoint $z_t$.}
The algorithm observes $f_t$ and suffers loss  $f_t(h_t)$.

The \emph{loss} and \emph{regret} of the algorithm on this particular input sequence are
\begin{align}
 &\text{Loss}_T = \sum_{t=1}^T f_t(h_t) . \label{eqn:define-loss} \\
 &\text{Regret}_T = \text{Loss}_T - \min_{h^*\in\H} \sum_{t=1}^T f_t(h^*) . \label{eqn:define-regret}
\end{align}
The regret objective is what one typically studies in adversarial settings, where we want to discount the loss incurred by the algorithm by the loss suffered by the best possible $h^*$ chosen with knowledge of the sequence of $f_t$'s. As we often consider randomized algorithms, we will generally consider \emph{expected} loss and regret, where the expectation is over any randomness in the algorithm not over the (possibly-randomized) input sequence of loss functions. An algorithm is said to guarantee regret $R(T)$ if the latter provides an upper bound on regret for every sequence of loss functions $f_1,\dots,f_T$.

We utilize the broad class of \emph{Follow-the-Regularized-Leader} (FTRL) online algorithms (Algorithm \ref{alg:ftrl}) \cite{zinkevich2003online,Shai12}.
Special cases of FTRL include Online Gradient Descent, Multiplicative Weights, and others.
Each FTRL algorithm is specified by a convex function $G : \H \to \R$ which is known as a \emph{regularizer} and is usually strongly convex with respect to some norm.
For example, Multiplicative Weights follows by using the negative entropy function as a regularizer, which is strongly-convex with respect to $\ell_1$ norm \citep{cesa2006prediction}.
Online Gradient Descent follows by using the regularizer $G(h) = \frac{1}{2}\|h\|_2^2$, which is strongly-convex with respect to $\ell_2$ norm.
These special cases have efficient closed-form solutions to the update rule for computing $h_{t+1}$.

\begin{algorithm}
 \KwIn{learning parameter $\eta$, convex regularizer $G: \H \to \R$}
 \For{$t = 1, \dots, T$}{
  post hypothesis $h_t$, observe loss function $f_t$\;
  update $\textstyle h_{t+1} = \inf_{h\in\H} \Big\{  \sum_{t'\leq t} f_{t'}(h) ~ + ~ \eta G(h)\Big\} $ 
  \;
 }
\caption{Follow-the-Regularized-Leader (FTRL).}
\label{alg:ftrl}
\end{algorithm}
It is well-known (and indeed follows as a special case of Lemma \ref{lemma:ftrl-iwtd-regret}) that, under the assumptions on our setting, FTRL algorithms guarantee an expected regret bound of $O(\sqrt{T})$, and this is tight with respect to $T$.

\subsection{Importance-Weighting Technique for Less Data} \label{subsec:importance-weighting}
As a starting point, suppose we wish to design an online learning algorithm that does not observe all of the arriving loss functions, but still performs well against the entire arrival sequence.

Because the arrival sequence may be adversarially chosen, a good algorithm should randomly choose to sample some of the arrivals.
In this section, we abstract away the decision of how to randomly sample.
(This will be the focus of Section \ref{sec:noregret-strategic}.)
In this section, we suppose that at each time $t$, after posting a hypothesis $h_t$, a probability $q_t > 0$ is specified by some external means as a (possibly random) function of the preceding time steps.
With probability $q_t$, we observe $f_t$; with probability $1-q_t$, we observe nil.

Our goal is to modify the FTRL algorithm for this setting and obtain a modified regret guarantee.
Notice crucially that the definition of loss and regret (\ref{eqn:define-regret}) are unchanged: We still suffer the loss $f_t(h_t)$ regardless of whether we observe $f_t$.

The key technique we use is \emph{importance weighting}.
The idea is that, if we only observe each of a sequence of values $x_i$ with probability $p_i$, then we can get an unbiased estimate of their sum by taking the sum of $\frac{x_i}{p_i}$ for those we do observe. To check this fact, let $\1_i$ be the indicator variable for the event that we observe $i$ and note that the expectation of our sum is $\E\left[ \sum_i \1_i \frac{x_i}{p_i} \right] = \sum_i x_i$.
This is called importance-weighting the observations (and is a specific instance of a more general machine learning technique).
Furthermore, if each $\frac{x_i}{p_i}$ is bounded and observed independently, we can expect the estimate to be quite good via tail bounds.

The importance-weighted modification to an online learning algorithm is outlined in Algorithm \ref{alg:iwtd-online}.
The importance-weighted regret guarantee we obtain is given in Lemma \ref{lemma:ftrl-iwtd-regret}.
It depends on the following key notation.
Our analysis and algorithm require a given norm $\| \cdot \|$, and we recall the definition of the dual norm $\| z \|_{\star} := \sup_{x; \|x \| \leq 1} x \cdot z$.
\begin{definition} \label{def:delta}
Given $h \in H$, and convex loss $f : \H \to \mathbb{R}$, let $\Delta_{h,f} := \|\nabla f(h) \|_{\star}$.
\end{definition}
We can informally think of $\Delta_{h,f}$ both as the ``difficulty'' of arrival $f$ when the current hypothesis is $h$, and as the ``value'' of observing $f$.
This interpretation is explored in Section \ref{sec:noregret-strategic} when we define the parameter $\algcost$.

\begin{algorithm}[h]
 \KwIn{access to Online Learning Algorithm (OLA)}
 \For{$t = 1, \dots, T$}{
  post hypothesis $h_t$ $\leftarrow$ OLA; observe sampling probability $q_t$\;
  toss $q_t$-weighted coin (Bernoulli sample) $\epsilon_t$ \;
  $
  \text{ if } \epsilon_t = 
  \begin{cases}
    1 & \text{input importance-weighted loss function} \quad \hat{f}_t(\cdot) = \frac{f_t(\cdot)}{q_t} \rightarrow  \text{OLA} \\
    0 & \text{input zero function} \quad \hat{f}_t(\cdot) \equiv 0 \rightarrow \text{OLA}
  \end{cases}$
 }
\caption{Importance-Weighted Online Learning Algorithm.}
\label{alg:iwtd-online}
\end{algorithm}

\def\lemmaftrliwtdregret{  
 Assume we implement Algorithm \ref{alg:iwtd-online} with nonzero sampling probabilities $q_1,\dots,q_T$. Assume the underlying OLA is FTRL (Algorithm~\ref{alg:ftrl}) with regularizer $G : \H \to \R$ that is strongly convex with respect to $\| \cdot \|$. Then the expected regret, with respect to the loss sequence $f_1, \ldots, f_T$, is no more than
  \[ R(T) = \frac{\beta}{\eta} + 2 \eta \E \left[ {\textstyle \sum_{t=1}^T \frac{\Delta_{h_t,f_t}^2}{q_t} }\right], \]
 where $\beta$ is a constant depending on $\H$ and $G$, $\eta$ is a parameter of the algorithm, and the expectation is over any randomness in the choices of $h_t$ and $q_t$.
} 
\begin{lemma} \label{lemma:ftrl-iwtd-regret}
\lemmaftrliwtdregret
\end{lemma}
We can recover the classic regret bound as follows: Take each $q_t = 1$, and note by the Lipschitz assumption that each $\Delta_{h_t,f_t} \leq 1$. Then by setting $\eta = \Theta(1/\sqrt{T})$, we get an expected regret bounded by $O(\sqrt{T})$.

\subsection{The ``Online-to-Batch'' Conversion}
So far so good: We can convert an online regret-minimization algorithm to use smaller amounts of data, and we postpone the question of how to price data till Section \ref{sec:noregret-strategic}. We now address the statistical learning problem, which is how to generate accurate \emph{predictions} based on the online learning process.

We address this with a standard tool known as the ``online-to-batch conversion,'' where we may leverage an online learning algorithm for use in a ``batch'' setting.
A sketch of this technique is as follows, and further details can be found in, e.g., \citet{Shai12}. Given a batch of i.i.d. data points, feed them one-by-one into the no-regret algorithm.
Because the algorithm has low regret, its hypotheses predicted well on average.
But since each data point was drawn i.i.d., this means that these hypotheses on average predict well on an i.i.d. draw from the distribution.
Thus it suffices to take the mean of the hypotheses to obtain low risk.

\begin{lemma}[Online-to-Batch \cite{cesa2004generalization}] \label{lemma:online-to-batch}
Suppose the sequence of convex loss functions $f_1,\dots,f_T$ are drawn i.i.d. from a distribution $\F$ and that an online learning algorithm with hypotheses $h_1,\dots,h_T$ achieves expected regret $R(T)$.
Let $\err(h) = \E_{f\sim\F} f(h)$ and $h^* = \arg\min_{h\in\H} \err(h)$.
For $\bar{h}_{1:T} = \frac{1}{T} \sum_{t=1}^T h_t$, we have
 \[ \E_{\substack{f_1,\dots,f_T, \\ \text{alg}}} \err(\bar{h}_{1:T}) \leq \err(h^*) ~ + ~ \frac{1}{T} R(T) .  \]
\end{lemma}

We note that this conversion will continue to hold in the data-purchasing no-regret setting we define next, since all that is required is that the algorithm output a hypothesis $h_t$ at each step and that there is a regret bound on these hypotheses.

%% file: noregret-strategic.tex
In this setting, we define the problem of regret minimization with purchased data.
We will design mechanisms with good regret guarantees for this problem, which will translate via the aforementioned online-to-batch conversion (Lemma \ref{lemma:online-to-batch}) into guarantees for our original problem of statistical prediction.

The essence of the data-purchasing no-regret learning setting is that an online algorithm (``mechanism'') is asked to perform well against a sequence of data, but by default, the mechanism does not have the ability to see the data.
Rather, the mechanism may purchase the right to observe data points using a limited budget.
The mechanism is still expected to have low regret compared to the optimal hypothesis in hindsight on the entire data sequence (even though it only observes a portion of the sequence).

\subsection{Problem Definition}
The \textbf{data-purchasing regret minimization problem} is parameterized by the hypothesis space $\H$, number of arriving data points $T$, and expected budget constraint $B$.
A \emph{problem instance} is a sequence of pairs $(c_1,f_1), \dots, (c_T, f_T)$ where each $f_t: \H \to \R$ is a convex loss function and each $c_t \in [0,1]$ is the cost associated with that data point.
We assume that the $f_t$ are $1$-Lipschitz, and let $\F$ be the set of such loss functions.

In this problem, we design a \emph{mechanism} implementing the operations ``post'' and ``receive'' and interacting with the problem instance as follows.
\begin{itemize}
  \item For each time step $t=1,\dots,T$:
  \begin{enumerate}
    \item The mechanism \emph{posts} a hypothesis $h_t$ and a pricing function $\pi_t: \F \to \R$, where $\pi_t(f)$ is the price posted for loss function $f$.
    \item Agent $t$ arrives, possessing $(c_t,f_t)$.
    \item If the posted price $\pi_t(f_t) \geq c_t$, then agent $t$ accepts the transaction: The mechanism pays $\pi_t(f_t)$ to the agent and \emph{receives} $(c_t, f_t)$. If $\pi_t(f_t) < c_t$, agent $t$ rejects the transaction and the mechanism \emph{receives} a null signal.
  \end{enumerate}
\end{itemize}
Note the key differences from the statistical learning setting: We must post a hypothesis $h_t$ at each time step (and we do not output a final prediction), and data is not assumed to come from a distribution.

The goal of the mechanism is to minimize the loss, namely $\sum_t f_t(h_t)$.
The definition of regret is also the same as in the classical setting (Equation \ref{eqn:define-regret}).
Note that we suffer a loss $f_t(h_t)$ at time $t$ regardless of whether we purchase $f_t$ or not.
The mechanism must also guarantee that, for every problem instance $(c_1,f_1),\dots,(c_T,f_T)$, it spends at most $B$ in expectation over its own internal randomness.

\subsection{The Importance-Weighting Framework}
Recall that, in Section \ref{subsec:importance-weighting}, we introduced the \emph{importance-weighting} technique for online learning. This gave regret guarantees for a learning algorithm when each arrival $f_t$ is observed with some probability $q_t$.

Our general approach will be to develop a strategy for randomly drawing posted prices $\pi_t$.
This will induce a probability $q_t$ of obtaining each arrival $f_t$.

Therefore, the entire problem has been reduced to choosing a posted-price strategy at each time step. This posted-price strategy should attempt to minimize the regret bound while satisfying the expected budget constraint.

A brief sketch of the proof arguments is as follows. After we choose a posted price strategy, each $q_t$ will be determined as a function of $h_t, c_t,$ and $f_t$.
($q_t$ is just equal to the probability that our randomly drawn price exceeds the agent's cost $c_t$.)
Thus, we can apply Lemma \ref{lemma:ftrl-iwtd-regret}, which stated that for these induced probabilities $q_t$, the expected regret of the learning algorithm is
 \[ \frac{\beta}{\eta} + 2 \eta \E \sum_t \frac{\Delta_{h_t,f_t}^2}{q_t}  , \]
where $\beta$ is a constant and $\eta$ is a parameter of the learning algorithm to be chosen later.

After we choose and apply such a strategy, the general approach to proving our regret bounds is to find an \emph{a priori} bound $M$ such that $2 \E \sum_t \frac{\Delta_{h_t,f_t}^2}{q_t} \leq M$.
Then the regret bound becomes $\frac{\beta}{\eta} + \eta M$.
If we know this upper-bound $M$ in advance using some prior knowledge, then we can choose $\eta = \Theta(1/\sqrt{M})$ as the parameter for our learning algorithms.
This gives a regret guarantee of $O(\sqrt{M})$.

\subsection{A First Step to Pricing: The ``At-Cost'' Variant} \label{subsec:noregret-at-cost}
The bulk of our analysis of the no-regret data-purchasing problem actually focuses on a slightly easier variant of the setting: If the arriving agent accepts the transaction, then the mechanism only has to pay the cost $c_t$ rather than the posted price $\pi_t(f_t)$.
We call this the ``at-cost'' variant of the problem.
This setting turns out to be much more analytically tractable: We derive optimal regret bounds for our mechanisms and matching lower bounds.
We then take the key approach and insights derived from this variant and apply them to produce a solution to the main no-regret data purchasing problem.
In order to keep the story moving forward, we summarize our results for the ``at-cost'' setting here and explore how they are obtained in Section \ref{sec:noregret-nonstrategic}.

In the at-cost setting, we are able to solve directly for the pricing strategy that minimizes the importance-weighted regret bound of Lemma \ref{lemma:ftrl-iwtd-regret}.
We first define one important quantity, then we state the strategy and result in Theorem \ref{theorem:noregret-nonstrategic}.
\begin{definition} \label{def:gamma}
For a fixed input sequence $(c_1,f_1),\dots,(c_T,f_T)$, $\Delta_{h,f}$ in Definition \ref{def:delta}, and a mechanism outputting (possibly random) hypotheses $h_1,\dots,h_T$, define
 \[  \algcost = \E \frac{1}{T} \sum_t \Delta_{h_t,f_t} \sqrt{c_t}  \]
where the expectation is over the randomness of the algorithm.
Note that $\algcost$ lies in $[0,1]$ by our assumptions on bounded cost and Lipschitz loss.
\end{definition}

Now we give the main result for the at-cost setting.

\def\theoremnoregretnonstrategic{ 
There is a mechanism for the ``at-cost'' problem of data purchasing for regret minimization that interfaces with FTRL and guarantees to meet the expected budget constraint, where for a parameter $\algcost \in [0,1]$ (Definition \ref{def:gamma}),
\begin{enumerate}
 \item The expected regret is bounded by $O\left(\max\left\{ \frac{T}{\sqrt{B}}\algcost ~,~ \sqrt{T} \right\}\right)$.
 \item This is optimal in that no mechanism can improve beyond constant factors.
 \item The pricing strategy is to choose a parameter $K = O\left(\frac{T}{B}\algcost\right)$ and draw $\pi_t(f)$ randomly according to a distribution such that $\Pr[\pi_t(f) \geq c] = \min\left\{ 1 ~,~ \frac{\Delta_{h_t,f}}{K \sqrt{c}} \right\}.$
\end{enumerate}
The only prior knowledge required is an estimate of $\algcost$ up to a constant factor.
} 
\begin{theorem} \label{theorem:noregret-nonstrategic}
\theoremnoregretnonstrategic
\end{theorem}

\subsection{Interpreting the Quantity $\algcost$}
\label{sec:gamma_interpretation}

Several of our bounds rely heavily on the quantity $\algcost$ which measures, in a sense, the ``financial difficulty'' of the problem.
We now devote some discussion to understanding $\algcost$ by answering four questions.

\textit{(1) How to interpret $\algcost$?}

$\algcost$ is an average, over time steps $t$, of $\Delta_{h_t,f_t} \cdot \sqrt{c_t}$.
Here, $\Delta_{h_t,f_t}$ intuitively captures both the ``difficulty'' of the data $f_t$ and also the ``value'' or ``benefit'' of $f_t$.
To explain the difficulty aspect, by examining the regret bound for FTRL learning algorithms (\emph{e.g.} the importance-weighted regret bound of Lemma \ref{lemma:ftrl-iwtd-regret} with all $q_t=1$), one observes that if each $\Delta_{h_t,f_t}$ is small, then we have an excellent regret bound for our learning algorithm; the problem is ``easy''.
To explain the value aspect, one can for concreteness take the Online Gradient Descent algorithm; the larger the gradient, the larger the update at this step, and $\Delta_{h_t,f_t}$ is the norm of the gradient.
And in general, the higher $\Delta_{h_t,f_t}$, the more likely we are to purchase arrival $f_t$.

Thus, $\algcost$ captures the correlations between the value of the arriving data and the cost of that data.
If either the mean of the costs or the average benefit $\Delta_{h_t,f_t}$ of the data is converging to $0$, then $\algcost \to 0$ and in these cases we can learn with high accuracy very cheaply, as may be expected.
More interestingly, it is possible to have both high average costs, and high average data-values, and yet still have $\algcost \to 0$ due to beneficial correlations.
In these cases we can learn much more cheaply than might be expected based on either the economic side or the learning side alone.

\vspace{5pt}
\textit{(2) When should we expect to have good prior knowledge of $\algcost$?}

Although in general $\algcost$ will be domain-specific, there are several reasons for optimism.
First, $\algcost$ compresses all information about the data and costs into a single scalar parameter (compare to the common mechanism-design assumption that the prior distribution of agents' values is fully known).
Second, we do not need very exact estimates of $\algcost$ (\emph{e.g.} we do not need to know $\algcost \pm \epsilon$): For order-optimal regret bounds, we only need an estimate within a constant factor of $\algcost$.
Third, $\algcost$ is directly proportional to $K$, which is a normalization constant in our pricing distribution: If we increase $K$, the probability of obtaining a given data point only decreases, and vice versa.
In fact, the best choice of $K$ is the normalization constant so that we run out of budget precisely when the last arrival leaves.
Thus, $K$ (equivalently, $\algcost$) can be estimated and adjusted online by tracking the ``burn rate'' (spending per unit time) of the algorithm.
\ifnum\arxiv=1  
In simulations,
\else  
In simulations (appearing in the full version of the paper),
\fi  
we have observed success with a simple approach of estimating $K$ based on the average correlation so far along with the burn rate, \emph{i.e.} if the current estimated $\algcost$ is $\hat{\algcost}$ and there are $\hat{T}$ steps remaining with $\hat{B}$ budget remaining to spend, set $K = \hat{\algcost}\hat{T}/\hat{B}$.

\vspace{5pt}
\textit{(3) What can we prove without prior knowledge of $\algcost$?}

It turns out that, if we only have an estimate of $\bar{c} = \frac{1}{T}\sum_t \sqrt{c_t}$, respectively $\mu = \frac{1}{T}\sum_t c_t$, then this suffices for regret guarantees on the order of $T\bar{c}/\sqrt{B}$, respectively $T\sqrt{\mu}/\sqrt{B}$.
This ``graceful degradation'' will continue to be true in the main setting.
The idea is that we can follow the optimal form of the pricing strategy while choosing any normalization constant $K \geq \frac{T}{B}\algcost$.
It may no longer be optimal, but it will ensure that we satisfy the budget and give guarantees depending on the magnitude of $K$.
So all we need is an approximate estimate of some value larger than $\algcost$.
Both $\bar{c}$ and $\mu$ are guaranteed to upper-bound on $\algcost$, so both can be used to pick $K$ while satisfying the budget.

To recap, knowledge of only a simple statistic such as the mean of the arriving costs suffices for good learning guarantees, with better knowledge translating to better guarantees.

\vspace{5pt}
\textit{(4) $\algcost$ depends on the algorithm---what are the implications?}

We first note that $\algcost$ can be upper-bounded by, for instance, $\sqrt{\mu}$ where $\mu$ is the average of the arriving costs.
So a bound containing $\algcost$ does imply nontrivial algorithm-independent bounds.
The purpose of $\algcost$ is to capture cases where we can do significantly better than such bounds because the algorithm is a good fit for the problem.
To see this, note that running the FTRL algorithm on the entire data sequence (with no budget constraint) gives a regret bound of $\frac{\beta}{\eta} + \eta\sum_{t=1}^T \Delta_{h_t,f_t}^2$.
The worst case has each $\Delta_{h_t,f_t}$ equal to $1$, producing a $\sqrt{T}$ regret bound.
But in a case where the algorithm has a small average $\Delta_{h_t,f_t}$ and the algorithm enjoys a better regret bound, we may also hope that this improvement is reflected in $\algcost$.

However, one might hope for an algorithm-independent quantity that, in analogy with VC-dimension, captures the ``difficulty'' of the purchasing and learning problem instance.
This leads to the question:

\textit{(4a) Can we remove the algorithm-dependence of the bound?}
One might hope to achieve a bound depending on an algorithm-independent quantity that captures correlations between data and cost.
A natural candidate is $\algcostopt := \frac{1}{T} \sum_t \Delta_{h^*,f_t} \sqrt{c_t}$.
In general, there are difficult cases where one can not achieve a bound in terms of $\algcostopt$.
However, in nicer scenarios we may expect $\algcost$ to approximate $\algcostopt$.
For instance, suppose $\loss(h,z) = \phi(h^\top z)$ where $\phi$ is a differentiable convex function whose gradient is $1$-Lipschitz --- commonly-used examples include the \emph{squared hinge loss} and the \emph{log loss}.
Under this condition, where again we are using $f_t(\cdot) := \ell(\cdot, z_t)$, we can show that
\begin{eqnarray*}
    \Delta_{h_t,f_t}\sqrt{c_t} - \Delta_{h^*,f_t}\sqrt{c_t}
    & = & \|\nabla \ell(h_t,z_t) \|_{\star} \sqrt{c_t} - \|\nabla \ell(h^*,z_t) \|_{\star} \sqrt{c_t} \\
    & \leq & \|(\phi'(h_t^\top z_t) - \phi'(h^{*\top} z_t)) z_t\|_{\star}  \\
    & \leq & |\phi(h_t^\top z_t) - \phi(h^{*\top} z_t)| = |\ell(h_t,z_t) - \ell(h^*,z_t)|.
\end{eqnarray*}
By the regret guarantee of our mechanism when run with a good algorithm, even initialized with very weak knowledge, this difference in losses per time step is $o(1)$, implying that $\algcost \to \algcostopt$.
A deeper investigation of this phenomenon is a good candidate for future work.

\subsection{Mechanisms and Results for Regret Minimization}
In the previous section, we presented our results for the easier ``at-cost'' variant.
We now apply the approach derived for that setting to the main regret minimization problem.

For this problem, unlike in the ``at-cost'' variant, we cannot in general solve for the form of the optimal pricing strategy.
This is intuitively because, when we must pay the price we post, the optimal strategy depends on $c_t$. But the algorithm cannot condition the purchasing decision directly on $c_t$, as this is private information of the arriving agent.

We propose simply drawing posted prices according to the optimal strategy derived for the at-cost setting, namely,
\begin{align}
	\Pr[\pi_t(f) \geq c] = \min\left\{ 1 ~,~ \frac{\Delta_{h_t,f}}{K \sqrt{c}} \right\} ,
\label{eqn:sampling-prob}
\end{align}
but with a different choice of normalization constant $K$.
We note that there is a pricing distribution that accomplishes this:
\begin{observation}
\label{observation:pricing-pdf}
For any $K$ and $\Delta_{h_t,f}$, there exists a pricing distribution on $\pi_t(f)$ that satisfies Equation \ref{eqn:sampling-prob}.
Letting $c^* = \Delta_{h_t,f}^2/K^2$, the CDF is given by $F(\pi) = \Pr[\pi_t(f) \leq \pi] = 0$ if $\pi \leq c^*$, $F(\pi) = 1 - \Delta_{h_t,f}/K\sqrt{\pi}$ if $c^* \leq \pi \leq 1$, and $F(\pi) = 1$ if $\pi > 1$.
\end{observation}

The pricing distribution is given in Figure \ref{fig:price-distribution}.
This strategy gives Mechanism \ref{mech:no-regret-unknown}.

\begin{figure}[h]
 \centering
 \subfloat[Probability density function of the pricing distribution. 
  The price $\pi(f) = 1$ with probability $\min\left\{ 1 , \Delta_{h_t,f}/K \right\}$.
  On the interval $(c^*,1)$ the density function is $x \mapsto \Delta_{h_t,f} / 2 K x^{3/2}$.]{
  \includegraphics[width=0.47\textwidth]{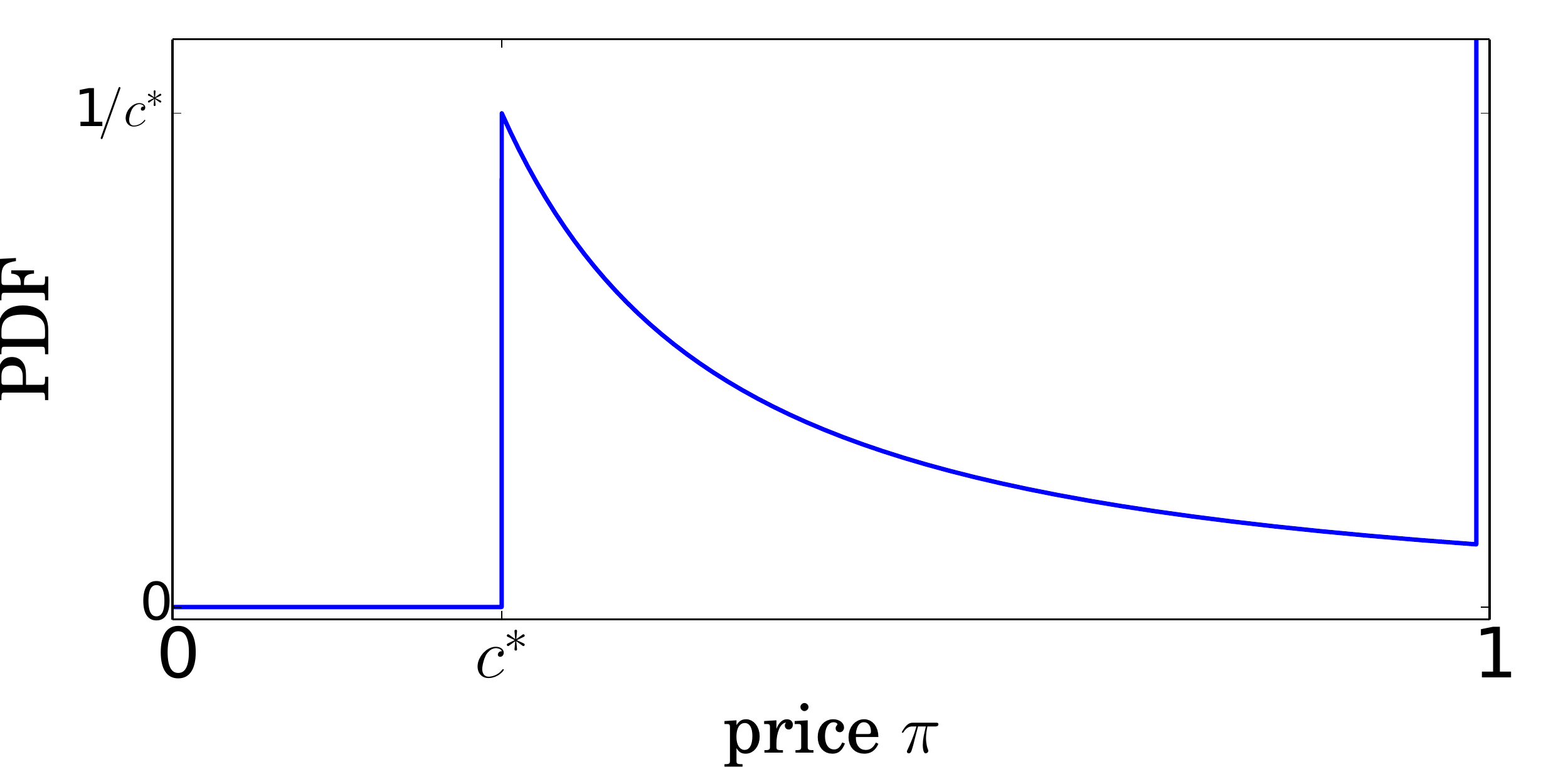}
 } 
 \hfill
 \subfloat[Cumulative distribution function of the pricing distribution. Equal to zero for $\pi \leq c^*$, then equal to $1 - \Delta_{h_t,f} / K\sqrt{\pi}$ on $(c^*, 1)$, then equal to $1$ at cost $1$.]{
  \includegraphics[width=0.47\textwidth]{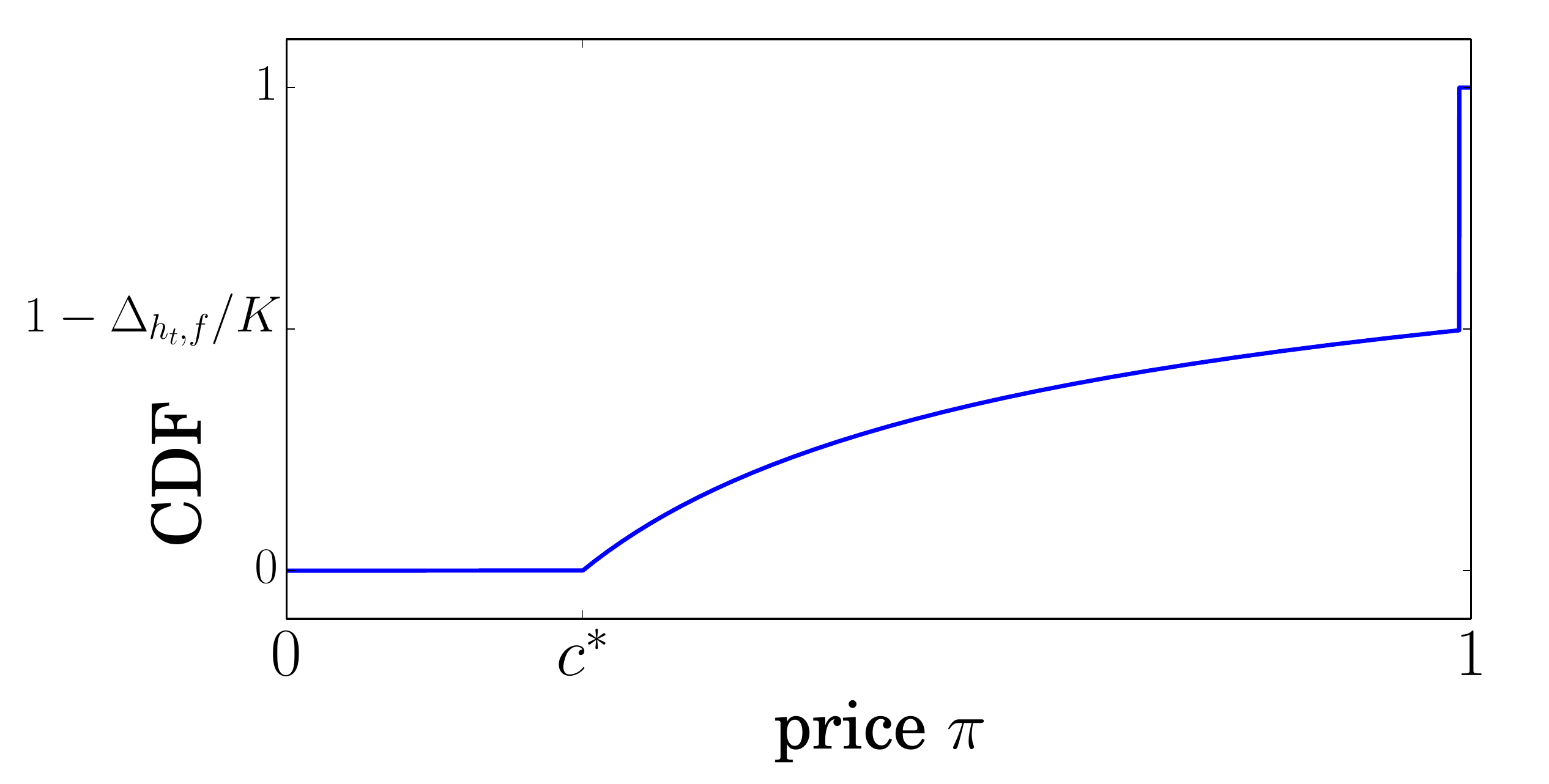}
 } 
 \caption{\textbf{The pricing distribution.}
  \small
  Illustrates the distribution from which we draw our posted prices at time $t$, for a fixed arrival $f$.
  The quantity $\Delta_{h_t,f}$ captures the ``benefit'' from obtaining $f$.
  $K$ is a normalization parameter.
  The distribution's support has a lowest price $c^*$, which has the form $c^* = \Delta_{h_t,f_t}^2 / K^2$.
 } 
 \label{fig:price-distribution}
\end{figure}

\begin{algorithm}
\SetAlgorithmName{Mechanism}{mechanism}{List of Mechanisms}
 \KwIn{parameters $K$, $\eta$, access to online learning algorithm (OLA)}
 set OLA parameter $\eta$\;
 \For{$t=1,\dots,T$}{
  post hypothesis $h_t$ $\leftarrow$ OLA\;
  post prices $\pi_t(f)$ drawn randomly such that
   $\Pr[\pi_t(f) \geq c] = \min\left\{ 1 ~,~ \frac{\Delta_{h_t,f}}{K \sqrt{c}} \right\} $ \;
  \eIf{we receive $(c_t, f_t)$}{
   let $q_t = \Pr_{\pi_t}[\pi_t(f_t) \geq c_t]$\;
   let \emph{importance-weighted loss function} $\hat{f}_t(\cdot) = \frac{f_t(\cdot)}{q_t}$\;
   send $\hat{f}_t$ $\rightarrow$ OLA\;
  }{ 
   send $0$ function $\rightarrow$ OLA\;
  } 
 }
\caption{Mechanism for no-regret data-purchasing problem.}
\label{mech:no-regret-unknown}
\end{algorithm}

As in the known-costs case, our regret bounds depend upon the prior knowledge of the algorithm.
It will turn out to be helpful to have prior knowledge about both $\algcost$ and the following parameter, which can be interpreted as $\algcost$ with all costs $c_t = 1$:
 \[ \algcostmax = \E \frac{1}{T} \sum_t \Delta_{h_t,f_t} . \]

\def\theoremnoregretunknownupper{ 
If Mechanism \ref{mech:no-regret-unknown} is run with prior knowledge of $\algcost$ and of $\algcostmax$ (up to a constant factor), then it can choose $K$ and $\eta$ to satisfy the expected budget constraint and obtain a regret bound of
 \[ O\left( \max\left\{ \frac{T}{\sqrt{B}}g ~,~ \sqrt{T} \right\} \right) , \]
where $g = \sqrt{\algcost \cdot \algcostmax}$ (by setting $K = \frac{T}{B}\algcostmax$).
Similarly, knowledge only of $\algcost$, respectively $\bar{c} = \frac{1}{T}\sum_t \sqrt{c_t}$, respectively $\mu = \frac{1}{T}\sum_t c_t$ suffices for the regret bound with $g = \sqrt{\algcost}$, respectively $g = \sqrt{\bar{c}}$, respectively $g = \mu^{1/4}$.
} 
\begin{theorem} \label{theorem:no-regret-unknown-upper}
\theoremnoregretunknownupper
\end{theorem}

We can observe a quantifiable ``price of strategic behavior'' in the difference between the regret guarantees of Theorems \ref{theorem:no-regret-unknown-upper} (this setting) and Theorem \ref{theorem:noregret-nonstrategic} (the ``at-cost'') setting:
 \[ \frac{T}{\sqrt{B}} \sqrt{\algcost \cdot \algcostmax}  ~~~ \text{vs} ~~~ \frac{T}{\sqrt{B}} \algcost  . \]
Note that $\algcostmax \geq \algcost$, and they approach equality as all costs approach the upper bound $1$, but become very different as the average cost $\mu \to 0$ while the maximum cost remains fixed at $1$.

\Paragraph{Comparison to lower bound}
Our lower-bound for the data purchasing regret minimization problem is $\Omega\left(\frac{T}{\sqrt{B}}\algcost\right)$ (follows from the lower bound for the at-cost setting, Theorem \ref{theorem:noregret-known-lower-gamma}). So the difference in bounds discussed above, a factor of $\sqrt{\algcostmax}$ versus $\sqrt{\algcost}$, is the only gap between our upper and lower bounds for the general data purchasing no regret problem.

The most immediate open problem in this paper is close this gap.
Intuitively, the lower bound does not take advantage of ``strategic behavior'' in that a posted-price mechanism may often have to pay significantly more than the data actually costs, meaning that it obtains less data in the long run.
Meanwhile, it may be possible to improve on our upper-bound strategy by drawing prices from a different distribution.

%% file: tieback.tex
In this section, we give the final mechanism, Mechanism \ref{mech:stat-learning}, for the data purchasing statistical learning problem.
The idea is to simply run the regret-minimization Mechanism \ref{mech:no-regret-unknown} on the arriving agents.
At each stage, Mechanism \ref{mech:no-regret-unknown} posts a hypothesis $h_t$.
We then aggregate these hypothesis by averaging to obtain our final prediction.

\begin{algorithm}
\SetAlgorithmName{Mechanism}{mechanism}{List of Mechanisms}
  \KwIn{parameters $K,\eta$, access to OLA}
  identify each data point $z$ with the loss function $f(\cdot) = \loss(\cdot, z)$\;
  run Mechanism \ref{mech:no-regret-unknown} with parameters $\eta, K$ and access to OLA\;
  let $h_1,\dots,h_T$ be the resulting hypotheses\;
  output $\bar{h} = \frac{1}{T}\sum_t h_t$\;
\caption{Mechanism for statistical learning data-purchasing problem.}
\label{mech:stat-learning}
\end{algorithm}

\begin{theorem} \label{thm:stat-learning-bounds}
Mechanism \ref{mech:stat-learning} guarantees spending at most $B$ in expectation and
 \[ \E \err(\bar{h}) \leq \err(h^*) + O\left({\textstyle \max\left\{\frac{g}{\sqrt{B}} ~,~ \sqrt{\frac{1}{T}}\right\} }\right) , \]
where $g = \sqrt{\algcost \cdot \algcostmax}$, assuming that $\algcost$ and $\algcostmax$ are known in advance up to a constant factor.

If one assumes approximate knowledge respectively of $\algcost$, of $\bar{c} = \frac{1}{T} \sum_t \sqrt{c_t}$, or of $\mu = \frac{1}{T}\sum_t c_t$, then the guarantee holds with respectively $g = \sqrt{\algcost}$, $g = \sqrt{\bar{c}}$, or $g = \mu^{1/4}$.
\end{theorem}
\begin{proof}
By Theorem \ref{theorem:no-regret-unknown-upper}, Mechanism \ref{mech:no-regret-unknown} guarantees an expected regret of $O\left(\max\left\{\frac{T}{\sqrt{B}}g ~,~ \sqrt{T} \right\}\right)$ when run with the specified prior knowledge for the specified values of $g$.
Therefore, the online-to-batch conversion of Lemma \ref{lemma:online-to-batch} proves the theorem.
\end{proof}

The statement of Main Result \ref{mainresult} is the special case where only $\algcost$ is known and $g = \sqrt{\algcost}$. 
A detailed discussion of $\algcost$ is in Section~\ref{sec:gamma_interpretation}.

%% file: noregret-nonstrategic.tex
In Section \ref{subsec:noregret-at-cost}, we stated our results for the easier at-cost variant of the regret minimization with purchased data problem.
This included the posted-price distribution that we use for our main results.
In this section, we show how these results and this distribution are derived.
The ``at-cost'' variant is formally defined in exactly the same way as the main setting, except that when $\pi_t \geq c_t$ and the transaction occurs, the mechanism only pays the cost $c_t$ rather than the posted price $\pi_t$.

We first show how our posted-price strategy is derived as the optimal solution to the problem of minimizing regret subject to the budget constraint.
The resulting upper bounds for the ``at-cost'' variant were given in Theorem \ref{theorem:noregret-nonstrategic}.
Then, we give some fundamental lower bounds on regret, showing that in general our upper bounds cannot be improved upon here.
These lower bounds also hold for the main no-regret data purchasing problem, where there is a small gap to the upper bounds.

\subsection{Deriving an Optimal Pricing Strategy}
We begin by asking what seems to be an even easier question. Suppose that for every pair $(c_t,f_t)$ that arrives, we could first ``see'' $(c_t,f_t)$, then choose a probability with which to obtain $(c_t,f_t)$ and pay $c_t$. What would be the optimal probability with which to take this data?
\def\lemmaknownoptimalstrategy{ 
  To minimize the regret bound of Lemma \ref{lemma:ftrl-iwtd-regret}, the optimal choice of sampling probability is of the form
  $\displaystyle q_t = \min\left\{ 1 ~,~ \Delta_{h_t,f_t}/K^* \sqrt{c_t} \right\}.$ 
  The normalization factor $K^* \approx \frac{T}{B}\algcost$.
} 
\begin{lemma} \label{lemma:known-optimal-strategy}
\lemmaknownoptimalstrategy
\end{lemma}
The proof follows by formulating the convex programming problem of minimizing the regret bound of Lemma \ref{lemma:ftrl-iwtd-regret} subject to an expected budget constraint.
It also gives the form of the normalization constant $K^*$, which depends on the input data sequence and the hypothesis sequence.

The key insight is now that we can actually achieve the sampling probabilities dictated by Lemma \ref{lemma:known-optimal-strategy} using a randomized posted-price mechanism.
Notice that these optimal sampling probabilities are decreasing in $c_t$.
In general, when drawing a price from some distribution, the probability that it exceeds $c$ will be decreasing in $c$.
So it only remains to find the posted-price distribution that actually induces the sampling probabilities that we want for all $c$ simultaneously.
That is, by randomly drawing posted prices according to our distribution, we choose to purchase $(c_t,f_t)$ with exactly the probability $q_t$ stated in Lemma \ref{lemma:known-optimal-strategy}, for any possible value of $c_t$ and without knowing $(c_t, f_t)$.

Thus, our final mechanism for the at-cost variant is to simply apply Mechanism \ref{mech:no-regret-unknown}, but only pay the cost of the arrival rather than the price we posted. We set $K = \frac{T}{B}\algcost$.
Note that this choice of normalization constant $K$ is different from the main setting because we on average pay less in the at-cost setting; this leads to the difference in the regret bounds.
Our main bound for the at-cost variant was given in Theorem \ref{theorem:noregret-nonstrategic}.
An open problem for this setting is whether one can obtain the same regret bounds without any prior knowledge at all about the arriving costs and data.

\subsection{Lower Bounds for Regret Minimization}
Here, we prove lower bounds analogous to the classic regret lower bound, which states that no algorithm can guarantee to do better than $O(\sqrt{T})$.
These lower bounds will hold even in the ``at-cost'' setting, where they match our upper bounds.
An open problem is to obtain a larger-order lower bound for the main setting where the mechanism pays its posted price.
This would show a separation between the at-cost variant and the main problem.

First, we give what might be considered a ``sample complexity'' lower bound for no-regret learning: It specializes our setting to the case where all costs are equal to one (and this is known to the algorithm in advance), so the question is what regret is achievable by an algorithm that observes $B$ of the $T$ arrivals.

\def\theoremnoregretknownlower{ 
  Suppose all costs $c_t = 1$. No algorithm for the at-cost online data-purchasing problem has regret better than $O(T/\sqrt{B})$; that is, for every algorithm, there exists an input sequence on which its regret is $\Omega(T/\sqrt{B})$.
} 
\begin{theorem} \label{theorem:noregret-known-lower}
\theoremnoregretknownlower
\end{theorem}
\begin{proofidea} We will have two coins, with probabilities $\frac{1}{2} \pm \epsilon$ of coming up heads.
We will take one of the coins and provide $T$ i.i.d. flips as the input sequence.
The possible hypotheses for the algorithm are $\{\text{heads}, \text{tails}\}$ and the loss is zero if the hypothesis matches the flip and one otherwise.
The cost of every data point will be one.

The idea is that an algorithm with regret much smaller than $T\epsilon$ must usually predict heads if it is the heads-biased coin and usually predict tails if it is the tails-biased coin.
Thus, it can be used to distinguish these cases.
However, there is a lower bound of $\Omega\left(\frac{1}{\epsilon^2}\right)$ samples required to distinguish the coins, and the algorithm only has enough budget to gain information about $O(B)$ of the samples.
Setting $\epsilon = 1/\sqrt{B}$ gives the regret bound.
\end{proofidea}

We next extend this idea to the case with heterogeneous costs.
The idea is very simple: Begin with the problem from the label-complexity lower bound, and introduce ``useless'' data points and heterogeneous costs.
The worst or ``hardest'' case for a given average cost is when cost is perfectly correlated with benefit, so all and only the ``useful'' data points are expensive.
\def\theoremnoregretknownlowergamma{ 
  No algorithm for the non-strategic online data-purchasing problem has expected regret better than $O\left(\algcost T / \sqrt{B}\right)$; that is, for every $\algcost$, for every algorithm, there exists a sequence with parameter $\algcost$ on which its regret is $\Omega\left(\algcost T / \sqrt{B} \right)$.
  Similarly, for $\bar{c} = \frac{1}{T}\sum_t \sqrt{c}$ and $\mu = \frac{1}{T}\sum_t c_t$, we have the lower bounds $\Omega\left(T \bar{c} / \sqrt{B}\right)$  and $\Omega\left(T \sqrt{\mu} / \sqrt{B}\right)$.
} 
\begin{theorem} \label{theorem:noregret-known-lower-gamma}
\theoremnoregretknownlowergamma
\end{theorem}

%% file: sim.tex
In this section, we give some examples of the performance of our mechanisms on data.
We use a binary classification problem with feature vector $x \in \mathbb{R}^d$ and label $y \in \{-1,1\}$.
The dataset is described in Figure \ref{fig:sim-dataset}.

\begin{figure}[h]
 \centering
 \subfloat[Visualizing the classification problem without costs.]{
  \includegraphics[width=0.35\textwidth]{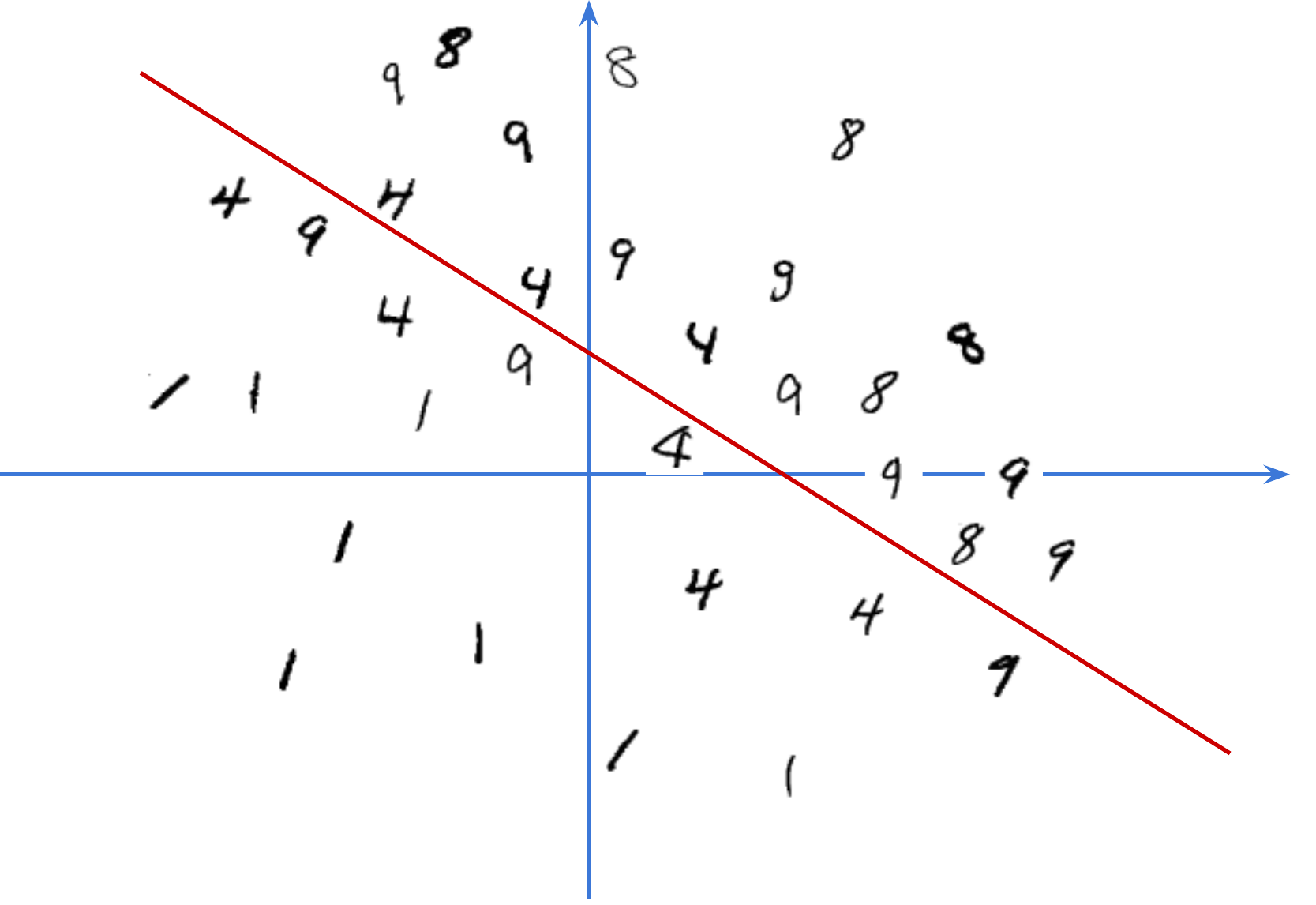}
 } 
 \hspace{6em}
 \subfloat[A brighter green background corresponds to a higher-cost data point.]{
  \includegraphics[width=0.35\textwidth]{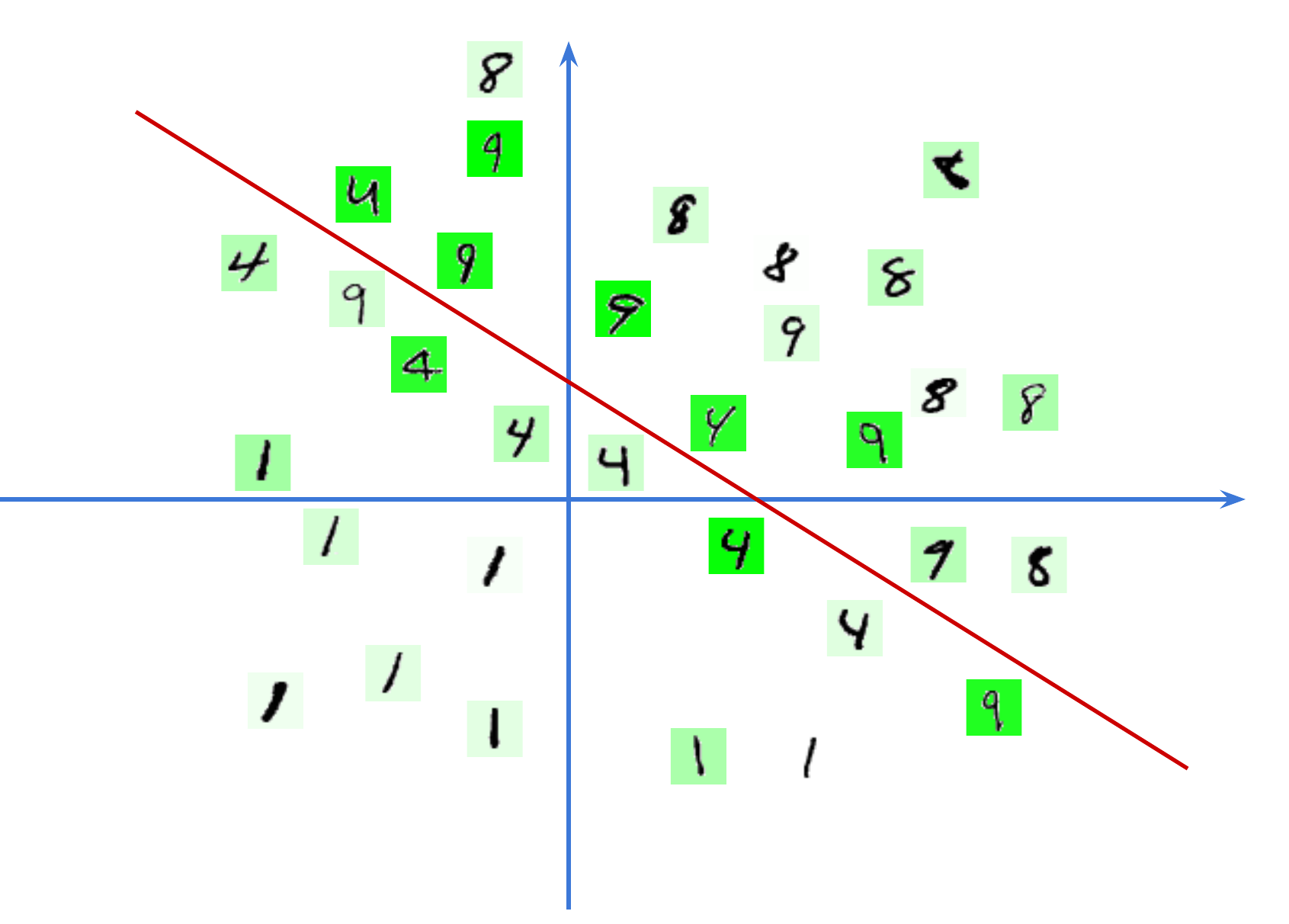}
 } 
 \caption{\textbf{Dataset.}
  \small
  Data points are images of handwritten digits, each data point consisting of a feature vector $x$ of grayscale pixels and a label $y$, the digit it depicts. We use the MNIST handwritten digit dataset (http://yann.lecun.com/exdb/mnist/).
 The algorithm is asked to distinguish between two ``categories'' of digits, where ``positive'' examples are digits $9$ and $8$ and ``negative'' examples are $1$ and $4$ (all other digits are not used). The number of training examples is $T = 8503$.
 This task allows us to adjust the correlations by drawing costs differently for different digits.}
 \label{fig:sim-dataset}
\end{figure}

\begin{figure}[h]
 \centering
 \subfloat[A comparison of mechanisms. ``Naive'' offers a maximum price of $1$ to every arrival until out of budget. ``Ours'' is Mechanism \ref{mech:stat-learning}, with $K$ initialized to $0$ and then adjusted online according to the estimated average $\algcost$ on the data so far. ``Baseline'' obtains every data point (has no budget constraint). Costs are distributed uniform $(0,1)$ independently. Each datapoint is an average of $4000$ trials, with standard error of at most $0.0002$.]{
  \includegraphics[width=0.47\textwidth]{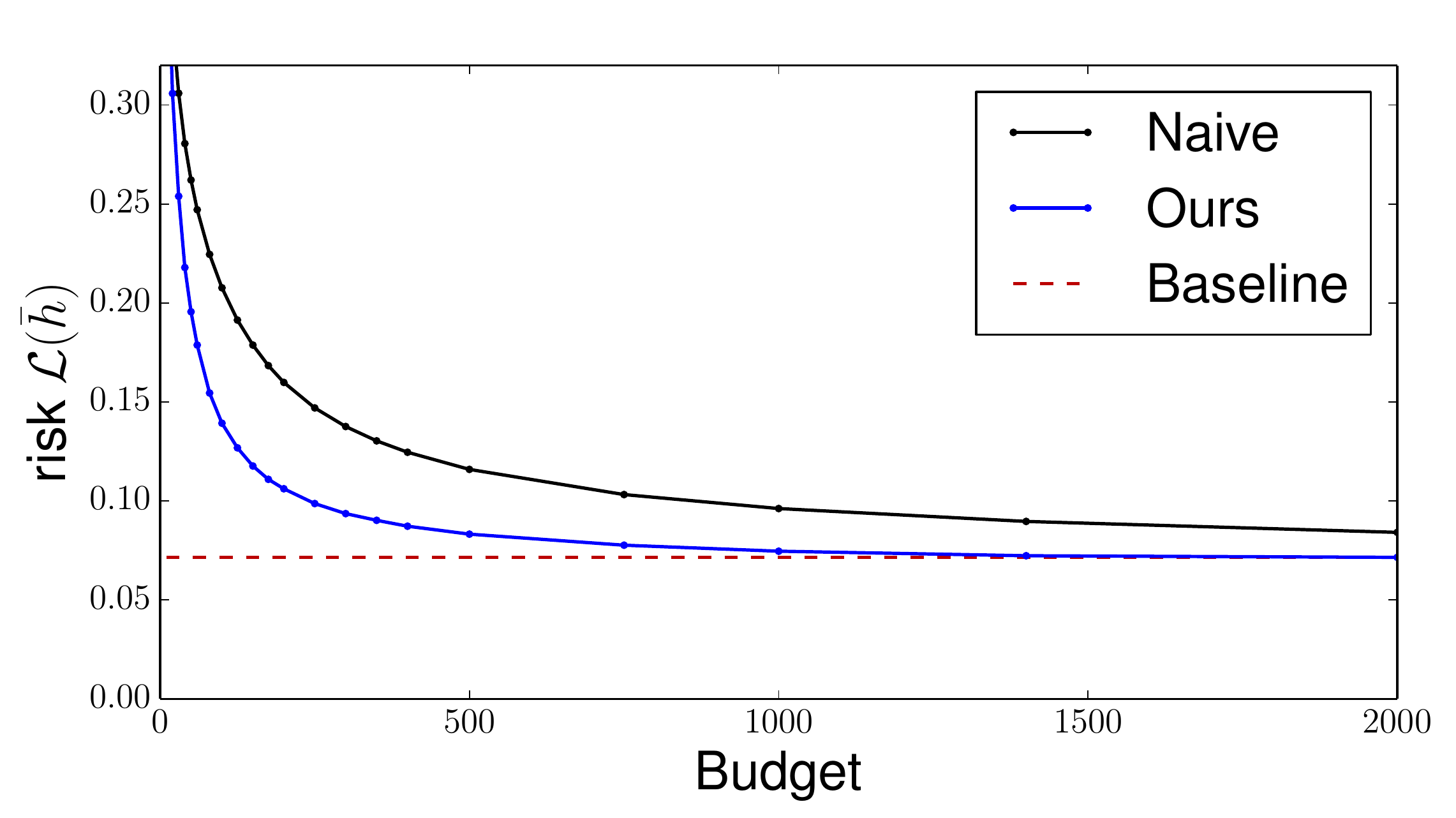}
 }
 \hfill
 \subfloat[An illustration of the role of cost-data correlations. The marginal distribution of costs is $1$ with probability $0.2$ and free otherwise, but the correlation of cost and data changes. The performance of Naive and the Baseline do not change with correlations. The larger-$\algcost$ case has high-cost points consisting of only $4$s and $9$s, while $\algcost$ is smaller when costs and data are independent. Each datapoint is an average of $2000$ trials, with standard error of at most $0.0004$.]{
  \includegraphics[width=0.47\textwidth]{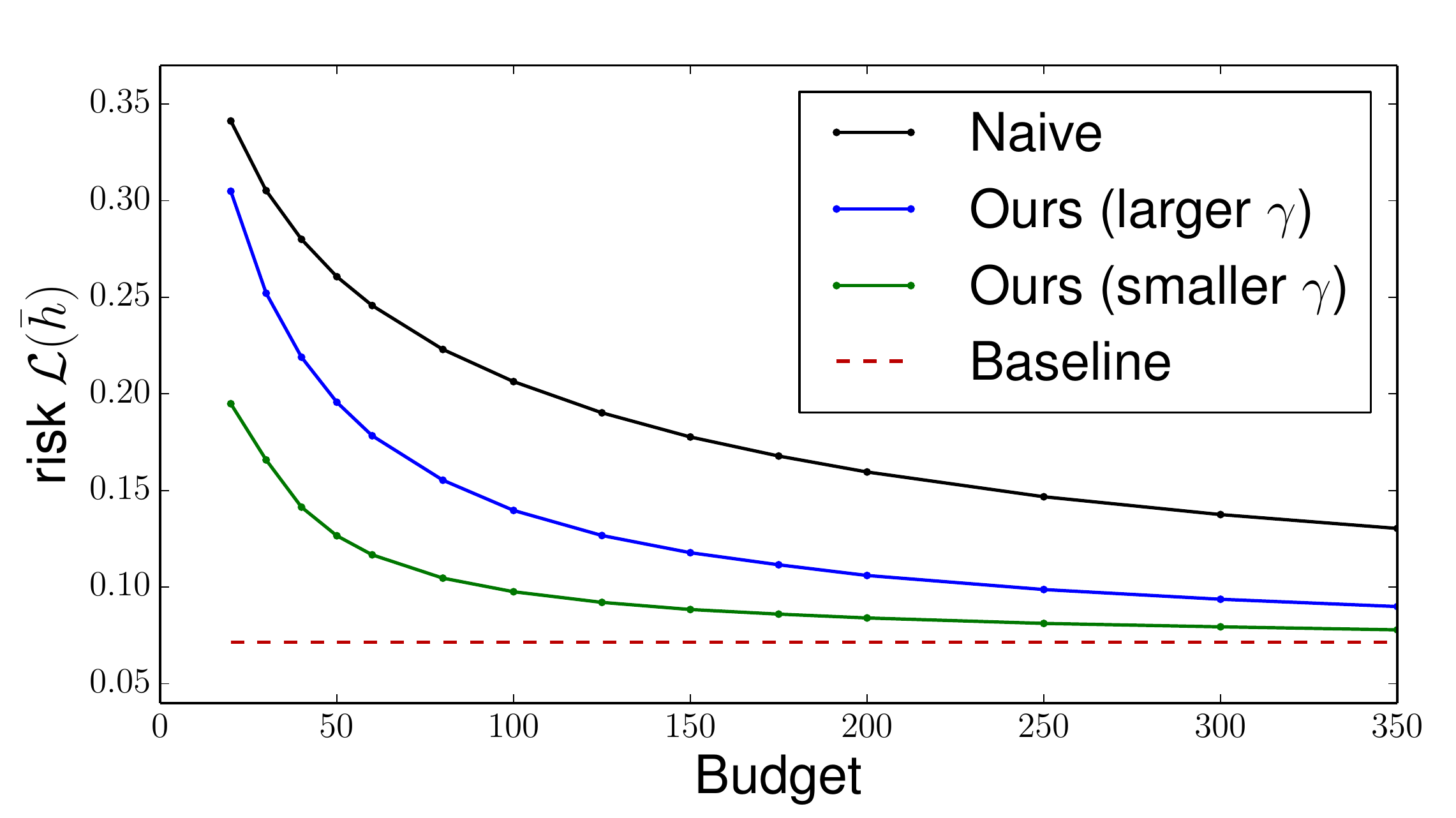}
 }
 \caption{\textbf{Examples of mechanism performance.}}
 \label{fig:sim-compare-prior}
\end{figure}

The hypothesis is a hyperplane classifier, \emph{i.e.} vector $w$ where the example is classified as positive if $w\cdot x \geq 0$ and negative otherwise; the risk is therefore the error rate (fraction of examples misclassified).
 For the implementation of the online gradient descent algorithm, we use a ``convexified'' loss function, the well-known hinge loss: $\loss(w,(x,y)) = \max\{0, 1 - y(w\cdot x)\}$ where $y \in \{-1,1\}$. 

In our simulations, we give each mechanism access to the exact same implementation of the Online Gradient Descent algorithm, including the same parameter $\eta$ chosen to be $0.1/c$ where $c$ is the average norm of the data feature vectors.
We train on a randomly chosen half of the dataset and test on the other half.

The ``baseline'' mechanism has no budget cap and purchases every data point. The ``naive'' mechanism offers a maximum price of $1$ for every data point until out of budget. ``Ours'' is an implementation of Mechanism \ref{mech:stat-learning}. We do not use any prior knowledge of the costs at all: We initialize $K=0$ and then adjust $K$ online by estimating $\algcost$ from the data purchased so far. (For a symmetric comparison, we do not adjust $\eta$ accordingly; instead we leave it at the same value as used with the other mechanisms.)
The examples are shown in Figure \ref{fig:sim-compare-prior}.

%% file: conclusion.tex
\subsection{Agent-Mechanism Interaction Model}
Our model of interaction, while perhaps the simplest initial starting point, involves some subtleties that may be interesting to address in the future.
A key property is that we need to obtain both an arriving agent's data point $z$ and her cost $c$.
The reason is that the cost is used to importance-weight the data based on the probability of picking a price larger than that cost. (The cost report is also required by \cite{RothSchoenebeck12} for the same reason.)
As discussed in Section \ref{sec:stat-learning}, a na\"{i}ve implementation of this model is incentive-compatible but not strictly so.
Exploring implementations, such as the trusted third party approach mentioned, is an interesting direction.
For instance, in a strictly truthful implementation, the arriving agent can cryptographically commit to a bid, \emph{e.g.} by submitting a cryptographic hash of her cost.
Then the prices are posted by the mechanism.
If the agent accepts, she reveals her data and her cost, verifying that the cost hashes to her commitment.
It is strictly truthful for the agent to commit to her true cost.

This paper focused on the learning-theoretic aspects of the problem, but exploring the model further or proposing alternatives is also of interest for future work.

\subsection{Conclusions and Directions}
The contribution of this work was to propose an \emph{active} scheme for learning and pricing data as it arrives online, held by strategic agents.
The active approach allows learning from past data and selectively pricing future data.
Our mechanisms interface with existing no-regret algorithms in an essentially black-box fashion (although the proof depends on the specific class of algorithms).
The analysis relies on showing that they have good guarantees in a model of no-regret learning with purchased data.
This no-regret setting may be of interest in future work, to either achieve good guarantees with no foreknowledge at all other than the maximum cost, or to propose variants on the model.

The no-regret analysis means our mechanisms are robust to adversarial input. But in nicer settings, one might hope to improve on the guarantees. One direction is to assume that costs are drawn according to a known marginal distribution (although the correlation with the data is unknown).
A combination of our approach and the posted-price distributions of \citet{RothSchoenebeck12} may be fruitful here.

Broadly, the problem of purchasing data for learning has many potential models and directions for study. One motivating setting, closer to crowdsourcing, is an active problem where data points consist of pairs (example, label) and the mechanism can offer a price for anyone who obtains the label of a given example. In an online arrival scheme, such a mechanism could build on the importance-weighted active learning paradigm~\citep{IWAL09}.

%% file: appendix.tex

\section{Tools for Converting Regret-Minimizing Algorithms}

\begin{lemma}[Lemma \ref{lemma:ftrl-iwtd-regret}]
\lemmaftrliwtdregret
\end{lemma}
\begin{proof}
\def\hstar{\mathbf{h^*}}
Let $\hstar = \inf_{h\in\H} \sum_t f_t(h)$.
We wish to prove that
 \[ \E_{\{h_t,q_t\}} \sum_t f_t(h_t) \leq \sum_t f_t(\hstar) ~+~ R \]
where $\{h_t,q_t\}$ is shorthand for $\{h_1,q_1,\dots,h_T,q_T\}$ and
 \[ R = \frac{\beta}{\eta} + 2 \eta \E_{\{h_t,q_t\}} \left[ \sum_t \frac{\Delta_{h_t,f_t}^2}{q_t} \right] . \]
As a prelude, note that in general these expectations could be quite tricky to deal with.
We consider a fixed input sequence $f_1,\dots,f_T$, but each random variable $q_t,h_t$ depends on the prior sequence of variables and outcomes.
However, we will see that the nice feature of the importance-weighting technique of Algorithm \ref{alg:iwtd-online} helps make this problem tractable.

Some preliminaries: Define the importance-weighted loss function at time $t$ to be the random variable
 \[ \hat{f}_t(h) = \begin{cases} \frac{f_t(h)}{q_t}  & \text{obtain $f_t$}  \\ 0  & \text{o.w.} \end{cases} \]
Let $\1_t$ be the indicator random variable equal to $1$ if we obtain $f_t$, which occurs with probability $q_t$, and equal to $0$ otherwise.
Then notice that for any hypothesis $h$,
\begin{align}
 \hat{f}_t(h) &= \1_t \frac{f_t(h)}{q_t}  \nonumber \\
 \implies \E_{\1_t} \left[ \hat{f}_t(h) \mid q_t \right]  &= f_t(h) . \label{eqn:iwtd-fhat}
\end{align}
To be clear, the expectation is over the random outcome whether or not we obtain datapoint $f_t$ conditioned on the value of $q_t$; and conditioned on the value of $q_t$, by definition we obtain datapoint $f_t$ with probability $q_t$ and obtain the $0$ function otherwise.

Now we proceed with the proof.
For any method of choosing $q_1,\dots,q_T$ and any resulting outcomes of $\1_t$, Algorithm \ref{alg:iwtd-online} reduces to running the Follow-the-Regularized-Leader algorithm on the sequence of convex loss functions $\hat{f}_1,\dots,\hat{f}_T$.
Thus, by the regret bound proof for FTRL (Lemma \ref{lemma:ftrl-regret}), FTRL guarantees that for every fixed ``reference hypothesis'' $\mathbf{h} \in \H$:
\begin{equation}
 \sum_t \hat{f}_t(h_t) \leq \sum_t \hat{f}_t(\mathbf{h}) ~ + ~ \hat{R}  \nonumber 
\end{equation}
where
\begin{align*}
 \hat{R} &= \frac{\beta}{\eta} ~ + ~ 2 \eta \sum_t \Delta_{h_t,\hat{f}_t}^2  \\
         &= \frac{\beta}{\eta} ~ + ~ 2 \eta \sum_t \1_t \frac{\Delta_{h_t,f_t}^2}{q_t^2} .
\end{align*}
(Recall that $\Delta_{h,f} = \|\nabla f(h)\|_{\star}$.)
Now we will take the expectation of both sides, separating out the expectation over the choice of $q_t$, over $h_t$, and over $\1_t$:
\begin{align*}
 \sum_t \E_{h_t,q_t} \left[ \E_{\1_t}\left[ \hat{f}_t(h_t) \mid h_t,q_t\right] \right]
  &\leq \sum_t \E_{h_t,q_t} \left[ \E_{\1_t} \left[ \hat{f}_t(\mathbf{h}) \mid h_t, q_t\right] \right]  ~ + ~ \E_{\{h_t,q_t\}}\left[\E_{\{\1_t\}} \left[ \hat{R} \mid \{h_t,q_t\}\right]\right] .
\end{align*}
Use the importance-weighting observation above (\ref{eqn:iwtd-fhat}):
\begin{align*}
 \E_{\{h_t,q_t\}} \sum_t f_t(h_t)
  &\leq \sum_t f_t(\mathbf{h}) ~ + ~ R
\end{align*}
where
 \[ R = \frac{\beta}{\eta} ~ + ~ 2 \eta \E_{\{h_t,q_t\}} \left[ \sum_t \frac{\Delta_{h_t,f_t}^2}{q_t} \right] . \]
In particular, because this holds for every reference hypothesis $\mathbf{h}$, it holds for $\mathbf{h^*}$.
\end{proof}


\begin{lemma} \label{lemma:ftrl-regret}
\ignore{\cj{Maybe cite some reference? }}
Let $G$ be 1-strongly convex with respect to some norm $\|\cdot\|$. The regret
of Follow-The-Regularized-Leader algorithm with regularizer $G$ and convex loss
functions $f_1,\dots,f_T$ can be bounded by
\[ 
   \frac{\beta}{\eta} + 2\eta \sum_t \Delta^2_{h_t,f_t} ~ ,
\]
where $\beta$ is the upper bound of $G(\cdot)$.
\end{lemma}
\begin{proof}
We reproduce the standard proof.
First, the regret of Follow-The-Regularized-Leader can be bounded by
\[
   \frac{1}{\eta} (R(h_T)-R(h_1)) + \sum_{t=1}^T (\loss(h_t,f_t)-\loss(h_{t+1},f_t)) .
\]

Below we show that $\loss(h_t,f_t)-\loss(h_{t+1},f_t) \leq 2\eta \|\nabla \loss(h_t,f_t)\|^2_\star$.

Define $\Phi_t(h)=R(h)/\eta + \sum_{i=1}^t \loss(h,f_i)$. 
By definition, we know $h_t = \arg\min_h \Phi_{t-1}(h)$. 
Since $\loss(\cdot)$ is convex and $R(\cdot)$ is 1-strongly convex, we know $\Phi_t(\cdot)$ is $(1/\eta)$-strongly convex for all $t$. Therefore, since $h_{t+1}$ minimizes $\Phi_t$, by definition of strong convex, we get
\[
   \Phi_t(h_t) \geq \Phi_t(h_{t+1}) + \frac{1}{2\eta} \| h_t-h_{t+1} \|^2
\]

After simple manipulations, we get
\begin{align*}
   \|h_t - h_{t+1}\|^2 & \leq 2\eta (\Phi_t(h_t) - \Phi_t(h_{t+1}))\\
	&= 2\eta (\Phi_{t-1}(h_t) - \Phi_{t-1}(h_{t+1})) + 2\eta(\loss(h_{t},f_t)-\loss(h_{t+1}, f_t))\\
	&\leq 2\eta(\loss(h_{t},f_t)-\loss(h_{t+1}, f_t))
\end{align*}
The last inequality comes from the fact that $h_t$ is the minimizer of $\Phi_{t-1}$. 

Since $\loss(\cdot)$ is convex, we have
\begin{align*}
   \loss(h_t, f_t) - \loss(h_{t+1}, f_t) 
	&\leq (h_t-h_{t+1}) \nabla \loss(h_t,f_t)\\
	&\leq \|h_t-h_{t+1}\| \|\nabla \loss(h_t,f_t)\|_\star
\end{align*}
The last inequality comes from the generalized Cauchy-Schwartz inequality.

Combining the above two inequalities together, we get
\begin{align*}
   \loss(h_t, f_t) - \loss(h_{t+1}, f_t) 
	\leq \|\nabla \loss(h_t,f_t)\|_\star \sqrt{2\eta(\loss(h_{t},f_t)-\loss(h_{t+1}, f_t))}
\end{align*}

By squaring and shifting sides,
\begin{align*}
   \loss(h_t, f_t) - \loss(h_{t+1}, f_t) 
	\leq 2\eta\|\nabla \loss(h_t,f_t)\|_\star^2 
\end{align*}
The proof is completed by inserting the inequality into the regret bound.
\end{proof}


\section{No regret ``at-cost'' setting}

\subsection{At-cost upper bounds}

\begin{lemma}[Lemma \ref{lemma:known-optimal-strategy}]
\lemmaknownoptimalstrategy
\end{lemma}
\begin{proof}
Recall that the regret bound of Lemma \ref{lemma:ftrl-iwtd-regret} is
 \[ \frac{\beta}{\eta} + 2\eta \E \sum_t \frac{\Delta_{h_t,f_t}^2}{q_t}  \]
where $q_t$ is the probability with which we choose to purchase arrival $(c_t,f_t)$.
We will solve for the choices of $q_t$ for each $t$.

Since $\beta$ is a constant and $\eta$ a parameter to be tuned later, our problem is to minimize the summation term in this regret bound.
This yields the following optimization problem:
\begin{align*}
   \min_{q_t} \sum_t \frac{\Delta_{h_t,f_t}^2}{q_t} &  \\
   \text{s.t.} \qquad \sum_t q_t \cdot c_t &\leq B  \\
   \qquad q_t &\leq 1  \qquad (\forall t) .
\end{align*}
The first constraint is the expected budget constraint, as we take each point $(c_t,f_t)$ with probability $q_t$ and pay $c_t$ if we do.
The second constrains each $q_t$ to be a probability.

To be completely formal, our goal is to minimize the expectation of the summation in the objective, as each $h_t$ and $q_t$ are random variables (they depend on the previous steps).
However, our approach will be to optimize this objective pointwise: For every prior sequence $h_1,\dots,h_t$ and $q_1,\dots,q_{t-1}$, we pick the optimal $q_t$. Therefore in the proof we will elide the expectation operators and argument. Similarly, since the budget constraint holds for all choices of $q_t$ that we make, we elide the expectation over the randomness in $q_t$.

The Lagrangian of this problem is
 \[ L(\lambda,\{q_t, \alpha_t\}) = \sum_t \frac{\Delta_{h_t,f_t}^2}{q_t} + \lambda \left(\sum_t q_t \cdot c_t ~ - ~ B\right) ~ + ~ \sum_t \alpha_t \left(q_t - 1\right) \]
with each $\lambda, q_t, \alpha_t \geq 0$. At optimum,
\begin{align*}
 0 &= \frac{\partial L}{\partial q_t}  \\
   &= -\frac{\Delta_{h_t,f_t}^2}{q_t^2} + \lambda c_t + \alpha_t ,
\end{align*}
implying that
 \[ q_t = \frac{\Delta_{h_t,f_t}}{\sqrt{\lambda c_t + \alpha_t}} . \]
By complementary slackness, $\alpha_t(q_t - 1) = 0$ at optimum, so consider two cases.
If $\alpha_t > 0$, then $q_t = 1$.
On the other hand, if $q_t < 1$, then $\alpha_t = 0$. Thus we may more simply write
 \[ q_t = \min\left\{ 1 ~,~ \frac{\Delta_{h_t,f_t}}{\sqrt{\lambda c_t}} \right\} . \]
Therefore, our normalization constant $K^* = \sqrt{\lambda}$.
To solve for $\lambda$, by complementary slackness, $\lambda\left(\sum_t q_t \cdot c_t ~ - ~ B\right) = 0$.
If $\lambda = 0$, then the form of $q_t$ and prior discussion implies that all $q_t = 1$, and we have $\sum_t c_t \leq B$; in other words, we have enough budget to purchase every point. Otherwise, the budget constraint is tight and $\sum_t q_t \cdot c_t = B$, so
 \[ \sum_t c_t \cdot \min\left\{ 1 ~,~ \frac{\Delta_{h_t,f_t}}{\sqrt{\lambda c_t}} \right\} = B . \]
Let us call those points that are taken with provability $q_t = 1$ ``valuable'' and the others ``less valuable'', and let $S$ be the set of less valuable points, $S = \{t : q_t < 1\}$.
Then we can rewrite as
 \[ \sum_{t\not\in S} c_t ~ + ~ \sum_{t\in S} \frac{\Delta_{h_t,f_t} \sqrt{c_t}}{\sqrt{\lambda}} = B , \]
so
 \[ K^* = \sqrt{\lambda} = \frac{1}{B - \sum_{t\not\in S} c_t}\sum_{t\in S} \Delta_{h_t,f_t}\sqrt{c_t} . \]
This completes the proof.
Let us make several final comments and observations, however.
First, if the budget is small relative to the amount of data, then with Lipschitz loss functions, no data points will be taken with probability $q_t = 1$, so $S$ will equal all of $T$.
In this case, the expectation of $K^*$ is exactly $\frac{T}{B}\algcost$, which is the meaning of our informal statement $K^* \approx \frac{T}{B}\algcost$.

Second, this $K^*$ is optimal ``pointwise'', in that it includes advance knowledge of which data points will be taken and which hypotheses will be posted.
However, notice that, to satisfy the budget constraint, it suffices to take the expectation and choose a normalization constant
 \[ K = \E\left[\frac{1}{B - \sum_{t\not\in S} c_t}\sum_{t\in S} \Delta_{h_t,f_t}\sqrt{c_t}\right] . \]
Third, as noted above, the extreme case is when all $q_t < 1$ and in this case the above $K = \frac{T}{B}\algcost$ exactly.
While this will not be ``as optimal'' for the specific random outcomes of this sequence, it will suffice to prove good upper bounds on regret. Furthermore, it holds that \emph{any} choice of $K \geq \frac{T}{B}\algcost$ satisfies the expected budget constraint; and (by setting $\eta$ as a function of $K$) suffices to prove an upper bound on regret.
\end{proof}

\begin{theorem}[Theorem \ref{theorem:noregret-nonstrategic}]
\theoremnoregretnonstrategic
\end{theorem}
\begin{proof}
The lower bound proof appears in Theorem \ref{theorem:noregret-known-lower-gamma}.

For the upper bound, we will give a more careful argument first, obtaining a more subtle bound capturing the two extremes in the regret bound as well as the spectrum in between.
We will then simplify to get the theorem statement.

First, note as pointed out in the proof of Lemma \ref{lemma:known-optimal-strategy} that choosing any $K \geq \frac{T}{B}\algcost \geq \E[K^*]$ satisfies the expected budget constraint, as each probability of purchase $q_t$ only decreases.
We now just need to show that if we know $\algcost$ to within a constant factor larger, \emph{i.e.} set $K = O\left(\frac{T}{B}\algcost\right)$ and $\eta$ appropriately, then we achieve the regret bound.

By Lemma \ref{lemma:ftrl-iwtd-regret}, for any choices of $q_t$ and the learning parameter $\eta$, the regret bound satisfies
\begin{equation}
 Regret \leq \frac{\beta}{\eta} + 2\eta \E \sum_t \frac{\Delta_{h_t,f_t}^2}{q_t}  \label{eqn:regret-bound-with-eta}
\end{equation}
where $\beta$ is a constant. Our strategy is to set
 \[ q_t = \min\left\{ 1 ~,~ \frac{\Delta_{h_t,f_t}}{K \sqrt{c_t}} \right\} . \]
Recall from the proof of Lemma \ref{lemma:known-optimal-strategy} that in the optimal solution there were in general ``valuable'' points for which the probability of purchase was $q_t = 1$ and ``less-valuable'' points where $q_t < 1$.
We had $S = \{t : q_t < 1\}$.
Thus the summation term in the regret bound becomes
\begin{equation}
 \E \sum_{t\not\in S} \Delta_{h_t,f_t}^2 ~ + ~ \E \sum_{t \in S} \Delta_{h_t,f_t} \sqrt{c_t} K . \label{eqn:regret-sum-bound}
\end{equation}
Before we prove the theorem statement, let us show how to achieve the more subtle bound.
So for the sake of this argument, let $\algcost(S) = \frac{1}{|S|} \E \sum_{t\in S} \Delta_{h_t,f_t} \sqrt{c_t}$.
Let $K_S$ approximate the more precise form derived in the proof of Lemma \ref{lemma:known-optimal-strategy}; that is,
 \[ K_S = O\left(\frac{|S|}{B - \sum_{t\not\in S} c_t} \algcost(S)\right) . \]
Then the summation term of the regret bound (Expression \ref{eqn:regret-sum-bound}) is at most a constant times
\begin{align}
  &\sum_{t\not\in S} \Delta_{h_t,f_t}^2 ~ + ~ \frac{|S|^2}{B - \sum_{t\not\in S}c_t} \algcost(S)^2  \nonumber \\
  &\leq T - |S| ~+~ + ~ \frac{|S|^2}{B - \sum_{t\not\in S}c_t} \algcost(S)^2   \label{eqn:exact-choice-eta}
\end{align}
as each $\Delta_{h_t,f_t} \leq 1$.
It remains to select the parameter $\eta$ to use for the learning algorithm and plug into the original regret bound, Expression \ref{eqn:regret-bound-with-eta}.
If the algorithm has an accurate estimate of $K_S$, $|S|$, and $\sum_{t\not\in S} c_t$, then it can set $\eta$ equal to the square root of one over Expression \ref{eqn:exact-choice-eta}.
(Note this may be achievable by tuning $\eta$ online as well, perhaps even with a theoretical guarantee.)
In this case, the regret bound is
 \[ Regret \leq O\left(\sqrt{T - |S| ~ + ~ \frac{|S|^2}{B - \sum_{t\not\in S}c_t}\algcost(S)^2}\right) . \]
Note that as $B \to 0$, $|S| \to T$, and as $B \to \sum_t c_t$, $|S| \to 0$.

Now let us actually prove the Theorem as stated. Let $\algcost = \sum_t \Delta_{h_t,f_t} \sqrt{c_t}$ and let $K = \frac{T}{B}\algcost$.
The summation term in the regret bound, Expression \ref{eqn:regret-sum-bound}, is upper-bounded by
\begin{align*}
 &T + (T\algcost) K  \\
 &= T + \frac{T^2}{B}\algcost^2
\end{align*}
using that $T\algcost \geq \sum_{t\in S}\Delta_{h_t,f_t} \sqrt{c_t}$ since it is a summation over more (positive) terms.
Now by Expression \ref{eqn:regret-sum-bound},
 \[ Regret \leq \frac{\beta}{\eta} + 2\eta\left(T + \frac{T^2}{B}\algcost^2\right) . \]
Setting
 \[ \eta = \Theta\left(1 / \max\left\{ \sqrt{T} ~,~ \frac{T}{\sqrt{B}}\algcost \right\}\right) \]
gives a regret bound of the order of $1/\eta$.
\end{proof}


\subsection{At-cost lower bounds}
\begin{theorem}[Theorem \ref{theorem:noregret-known-lower}]
\theoremnoregretknownlower
\end{theorem}
\begin{proof}
Consider two possible input distributions: i.i.d. flips of a coin that has probability $\frac{1}{2} + \epsilon$ of heads, or of one with probability $\frac{1}{2} - \epsilon$.

It will suffice to prove the following:

\textbf{Claim 1:} If there is an algorithm with budget $B$ and expected regret at most $T\epsilon/6$, then there is an algorithm to distinguish whether a coin is $\epsilon$-heads-biased or $\epsilon$-tails-biased with probability at least $2/3$ using $18B$ coin flips.

This claim implies the theorem because it is known that distinguishing these coins requires $\Omega\left(1/\epsilon^2\right)$ coin flips; in other words, it implies that $\epsilon \geq \Omega\left(1/\sqrt{B}\right)$, so the algorithm's expected regret must be $\Omega\left(T/\sqrt{B}\right)$.

We prove Claim $1$ by proving the following two claims:

\textbf{Claim 2:} If an algorithm's expected regret is at most $T\epsilon/6$, then under the $\epsilon$-heads-biased coin, with probability at least $5/6$, it outputs the heads hypothesis more times than the tails hypothesis.
(And symmetrically under the tails-biased coin.)

\textbf{Claim 3:} An algorithm in this coin setting with budget $B$ can, with probability at least $5/6$, be simulated for $T$ rounds using at most $18B$ coin flips -- in the sense that its behavior is identical to its behavior on a full sequence of $T$ coin flips.

\textit{Proof of Claim 1 from 2 and 3.} We will take an algorithm with budget $B$ and regret $T\epsilon$ and use it to distinguish the coin using $18B$ coin flips: Using Claim 3, we can simulate the algorithm's behavior for all $T$ rounds using at most $18B$ coin flips, except with probability $1/6$.
Then, if the algorithm used the hypothesis heads more times than tails, we guess that the coin is heads-biased, and symmetrically.
By Claim 2, our guess is correct except with probability $1/6$.
By a union bound, therefore, this procedure correctly distinguishes the coin except with probability $1/3$, proving Claim 1.

\textit{Proof of Claim 2.} Suppose the coin being flipped is the heads-biased coin; everything that follows will hold symmetrically for the tails-biased coin.
Now, suppose that the algorithm outputs the hypothesis tails for $M$ of the $T$ rounds.
Since each round is an independent coin toss, if the hypothesis is tails then its expected loss on that round is $\frac{1}{2} + \epsilon$; if heads, $\frac{1}{2} - \epsilon$.
This gives an expected loss of $M\left(\frac{1}{2} + \epsilon\right) + (T-M)\left(\frac{1}{2} - \epsilon\right) = \frac{T}{2} + (2M-T)\epsilon$.

Meanwhile, the expected loss of the optimal hypothesis is at most $T\left(\frac{1}{2} - \epsilon\right)$, since this is the expected loss of the heads hypothesis.
Therefore, the algorithm's expected regret, if it outputs the hypothesis tails $M$ times on average, is at least
 \[ \frac{T}{2} + (2\E M-T)\epsilon - T\left(\frac{1}{2} - \epsilon\right) = 2\E M \epsilon . \]

If the algorithm's regret is at most $T\epsilon/6$, then this implies that $2\E M \epsilon \leq T\epsilon/6$, or $\E M \leq T/12$.
Thus by Markov's inequality, the probability that half or more of the hypotheses are tails is bounded by
\begin{align*}
 \Pr[M \geq T/2] &\leq \frac{\E M}{T/2}  \\
  &\leq 1/6 .
\end{align*}.

\textit{Proof of Claim 3.}
Here, we assume that $\epsilon < 1/6$, or $B$ is larger than a (relatively small) constant.

On each data point, there are four possible menus: whether to buy or not to buy if the point is a heads or is a tails.\footnote{The algorithm may make this a randomized menu, but we can simply consider the outcome of that random menu.}
If the menu is (don't buy, don't buy), then no coin flip is needed (the behavior of the algorithm is identical whether the coin is actually flipped or not).
Otherwise, the coin must be flipped, but the algorithm buys the data point with probability at least $\frac{1}{2}-\epsilon \geq \frac{1}{3}$ (the lowest probability of the remaining three menus).
Thus the expected number of flips needed before the budget is exhausted is at most $3B$, and by Markov's inequality, the probability that it exceeds $18B$ is at most $1/6$.
\end{proof}

\begin{theorem}[Theorem \ref{theorem:noregret-known-lower-gamma}]
\theoremnoregretknownlowergamma
\end{theorem}
\begin{proof}
We reduce to the previous theorem.
Consider the following distribution on input sequences.
There are three possible data points: heads, tails, and ``no coin''.
There are still two hypotheses, heads and tails. Both have loss $1$ on the ``no coin'' data point.

Now fix any $\algcost \in [0,1]$.
We will first send $(1-\algcost)T$ data points, all of which are ``no coin''.
The loss of either hypothesis on all of these points is $1$, and the cost of these points is zero.
Then, we will choose either the $\epsilon$-heads-biased or $\epsilon$-tails-biased coin, with $\epsilon = 1/\sqrt{B}$, and send $T' = \algcost T$ coin flips, just as in the previous proof.

Because the first $(1-\algcost)T$ points are irrelevant to the regret, the regret of any algorithm is simply its regret on these final $T'$ data points, which by the previous proof is at least on the order of $T'\epsilon = T'/\sqrt{B} = \algcost T/\sqrt{B}$.

Now to check that the parameter $\algcost$ chosen above really is the $\algcost$ value of the data sequence, note that the convexified hypothesis space for this problem is the space of distributions $p \in \mathbb{R}^2$ on $\{heads, tails\}$, with loss $1 - p\cdot(1,0)$ if the coin is heads or $1-p\cdot(0,1)$ if the coin is tails.
The gradient of the loss on either point for all $p$ is $(1,0)$ or $(0,1)$ respectively, and both have norm $1$.
So $\Delta_{h_t,f_t} = 1$ for all ``heads'' and ``tails'' data points.
Thus we have that $\frac{1}{T} \sum_t \Delta_{h_t,f_t} \sqrt{c_t} = \frac{T'}{T} = \algcost$.

Finally, noting that $\algcost = \bar{c}$ in this case gives the bound containing $\bar{c}$. For the lower bound with $\mu$, take the exact construction in Theorem \ref{theorem:noregret-known-lower} and let each point have $c_t = \mu$ instead of $c_t = 1$.
\end{proof}

\section{No regret --- main setting}

\begin{theorem}[Theorem \ref{theorem:no-regret-unknown-upper}]
\theoremnoregretunknownupper
\end{theorem}
\begin{proof}
The proof will proceed by finding a close-to-optimal value $K$ of the normalizing constant by considering the budget constraint, then plugging this into the regret term to get a bound.
The constant maximum price plays into this proof in a slightly non-obvious way.
Because of this, instead of setting this maximum price equal to $1$, we consider the generalization where costs may lie in $[0,c_{max}]$.

Consider time $t$ when $(c_t, f_t)$ arrives.
Recall that the approach at time $t$ is to draw a price for $f_t$ from the distribution where
 \[ A_t(c) = \Pr[\text{price} \geq c] = \min\left\{ 1 ~,~ \frac{\Delta_{h_t,f_t}}{K \sqrt{c}} \right\}. \]
Consider then the induced posted-price distribution, which is pictured in Figure \ref{fig:price-distribution}.
It has a point mass at $c_{max}$ of probability\footnote{If this quantity is greater than $1$, then we post a price of $c_{max}$ for this datapoint, and what follows will only be a looser upper bound on the amount spent.} $\Delta_{h_t,f_t} / K \sqrt{c_{max}}$.
Otherwise, it is continuous on the interval $[c^*, c_{max}]$ with density
 \[ -A_t'(\pi) = \frac{ \Delta_{h_t,f_t} }{ 2 K \pi^{3/2}} , \]
and the lower endpoint $c^*$ satisfies $A_t(c^*) = 1$, \emph{i.e.} $c^* = \Delta_{h_t,f_t}^2 / K^2$.

We first find the bound on $K$ such that the expected budget constraint is satisfied.
The expected amount spent on arrival $t$ can be computed as follows.
\begin{align*}
 &c_{max} \Pr[\text{price } = c_{max}] ~ + ~ \int_{\max\{c_t,c^*\}}^{c_{max}} x \left(\text{pdf at $x$}\right) dx  \\
 &= c_{max} \frac{\Delta_{h_t,f_t}}{K\sqrt{c_{max}}} + \int_{\max\{c_t,c^*\}}^{c_{max}} x \frac{\Delta_{h_t,f_t}}{2 K x^{3/2}} dx  \\
 &= \frac{\Delta_{h_t,f_t}}{K}\left(\sqrt{c_{max}} + \int_{\max\{c_t,c^*\}}^{c_{max}} \frac{1}{2\sqrt{x}} dx \right) \\
 &= \frac{\Delta_{h_t,f_t}}{K}\left(2\sqrt{c_{max}} - \sqrt{\max\{c_t,c^*\}}\right) .
\end{align*}

Now let $c_t^*$ be the value of $c^*$ for arrival $t$ (to distinguish its value in different timesteps).
By the budget constraint, we need to pick $K$ so that
 \[ \sum_t \E\left[\text{spend on arrival $(c_t, f_t)$}\right] \leq B , \]
so
 \[ \E \sum_t \frac{\Delta_{h_t,f_t}}{K}\left(2\sqrt{c_{max}} - \sqrt{\max\{c_t,c^*\}}\right) \leq B . \]
Now we make a simplification: If we substitute $c_t$ in for $\max\{c_t,c^*\}$, then the left-hand side only increases. Thus, to satisfy the previous inequality, it suffices to choose $K$ to satisfy
 \[ \E \sum_t \frac{\Delta_{h_t,f_t}}{K}\left(2\sqrt{c_{max}} - \sqrt{c_t}\right) \leq B . \]
Thus, we let
 \[ K_{min} = \E \frac{1}{B} \sum_t \Delta_{h_t,f_t}\left(2\sqrt{c_{max}} - \sqrt{c_t}\right) . \]
Recall our definition of the ``difficulty-of-the-input'' parameter
 \[ \algcost = \E \frac{1}{T} \sum_{t} \Delta_{h_t,f_t} \sqrt{c_t} , \]
and let
 \[ \algcostmax = \E \frac{1}{T} \sum_{t} \Delta_{h_t,f_t}\sqrt{c_{max}} . \]
Then we have
 \[ K_{min} = \frac{T}{B} \left(2\algcostmax - \algcost\right) . \]
We now have the setup to quickly derive bounds such as the theorem statements.
Note that any choice of $K \geq K_{min}$ satisfies the expected budget constraint.

For the first regret bound, suppose that we know both $\algcost$ and $\algcostmax$ up to a constant factor.
Then we can set $K = O(K_{min})$.
By Lemma \ref{lemma:ftrl-iwtd-regret}, the expected regret is bounded by
 \[ Regret \leq \frac{\beta}{\eta} + \eta \sum_t \frac{\Delta_{h_t,f_t}^2}{A_t(c_t)} \]
where $\beta$ is a constant and $\eta$ will be chosen later.

As in the known-costs scenario, let us split into those arrivals that we purchase with probability $1$ (this corresponds to $c_t < c_t^*$) and the others, letting $S = \{t : A_t(c_t) < 1\}$.
Then the summation term in the regret bound is bounded by a constant times
\begin{align}
 &\sum_{t\not\in S} \Delta_{h_t,f_t}^2 ~ + ~ \sum_{t\in S} \Delta_{h_t,f_t}\sqrt{c_t} K_{min}  \nonumber \\
 &\leq T ~ + ~ \frac{T^2}{B} \algcost \left(2\algcostmax - \algcost\right) \label{eqn:strategic-regret-upper-expression}
\end{align}
where we have used the Lipschitz assumption on the loss function $\Delta_{h_t,f_t} \leq 1$.

As $\algcostmax \geq \algcost$, we do not lose much by taking the upper bound
\begin{equation}
 M = T ~ + ~ 2\frac{T^2}{B} \algcost \cdot \algcostmax . \label{eqn:noregret-strategic-M-upper}
\end{equation}
Now we can choose $\eta = \Theta(1/M)$ and obtain our regret bound of
\begin{align*}
 Regret &\leq O\left(\sqrt{M}\right)  \\
        &\leq O\left(\max\left\{ \sqrt{T} ~,~ \frac{T}{\sqrt{B}}\sqrt{\algcost \cdot \algcostmax} . \right\} \right).
\end{align*}
The other regret bounds will all follow by (1) upper-bounding $\algcostmax \leq \sqrt{c_{max}}$; (2) letting $K = \frac{T}{B}\sqrt{c_{max}}$; (3) upper-bounding $\algcost$; and (4) setting $\eta$ appropriately.
Note that this can only increase $K$, so the expected budget constraint is still satisfied.
The modifications simply give a different bound in Expression \ref{eqn:noregret-strategic-M-upper}, from which the rest of the argument follows analogously.

From (1) and (2), Expression \ref{eqn:noregret-strategic-M-upper} becomes
 \[ M = T ~ + ~ 2\frac{T^2}{B} \algcost \sqrt{c_{max}} . \]
First, if we know $\algcost$, then picking $\eta = \Theta(1/M)$ gives the corresponding bound.

Second, with only knowledge of $\bar{c} = \frac{1}{T} \sum_t \sqrt{c_t}$, observe that $\algcost \leq O(\bar{c})$ and plug in.

Third, observe that by Jensen's inequality $\bar{c} \leq \sqrt{\mu}$ (where $\mu = \frac{1}{T}\sum_t c_t$) and plug in.
\end{proof}

%% file: main.bbl
\begin{thebibliography}{20}
\providecommand{\natexlab}[1]{#1}
\providecommand{\url}[1]{\texttt{#1}}
\expandafter\ifx\csname urlstyle\endcsname\relax
  \providecommand{\doi}[1]{doi: #1}\else
  \providecommand{\doi}{doi: \begingroup \urlstyle{rm}\Url}\fi

\bibitem[Anthony and Bartlett(2009)]{anthony2009neural}
Martin Anthony and Peter~L Bartlett.
\newblock \emph{Neural network learning: Theoretical foundations}.
\newblock Cambridge University Press, 2009.

\bibitem[Balcan et~al.(2006)Balcan, Beygelzimer, and Langford]{BBL06}
Maria-Florina Balcan, Alina Beygelzimer, and John Langford.
\newblock Agnostic active learning.
\newblock In \emph{Proceedings of the 23rd International Conference on Machine
  Learning (ICML-06)}, 2006.

\bibitem[Balcan et~al.(2010)Balcan, Hanneke, and Vaughan]{balcan2010true}
Maria-Florina Balcan, Steve Hanneke, and Jennifer~Wortman Vaughan.
\newblock The true sample complexity of active learning.
\newblock \emph{Machine learning}, 80\penalty0 (2-3):\penalty0 111--139, 2010.

\bibitem[Beygelzimer et~al.(2009)Beygelzimer, Dasgupta, and Langford]{IWAL09}
Alina Beygelzimer, Sanjoy Dasgupta, and John Langford.
\newblock Importance weighted active learning.
\newblock In \emph{26th International Conference on Machine Learning
  (ICML-09)}, 2009.

\bibitem[Beygelzimer et~al.(2010)Beygelzimer, Hsu, Langford, and Zhang]{BHLZ10}
Alina Beygelzimer, Daniel Hsu, John Langford, and Tong Zhang.
\newblock Agnostic active learning without constraints.
\newblock In \emph{Advances in Neural Information Processing Systems
  (NIPS-10)}, 2010.

\bibitem[Cesa-Bianchi and Lugosi(2006)]{cesa2006prediction}
Nicolo Cesa-Bianchi and Gabor Lugosi.
\newblock Prediction, learning, and games.
\newblock 2006.

\bibitem[Cesa-Bianchi et~al.(2004)Cesa-Bianchi, Conconi, and
  Gentile]{cesa2004generalization}
Nicolo Cesa-Bianchi, Alex Conconi, and Claudio Gentile.
\newblock On the generalization ability of on-line learning algorithms.
\newblock \emph{Information Theory, IEEE Transactions on}, 50\penalty0
  (9):\penalty0 2050--2057, 2004.

\bibitem[Cummings et~al.(2015)Cummings, Ligett, Roth, Wu, and
  Ziani]{cummings2015accuracy}
Rachel Cummings, Katrina Ligett, Aaron Roth, Zhiwei~Steven Wu, and Juba Ziani.
\newblock Accuracy for sale: Aggregating data with a variance constraint.
\newblock In \emph{Proceedings of the 2015 Conference on Innovations in
  Theoretical Computer Science}, pages 317--324. ACM, 2015.

\bibitem[Dekel et~al.(2008)Dekel, Fischer, and Procaccia]{dekel2008incentive}
Ofer Dekel, Felix Fischer, and Ariel~D Procaccia.
\newblock Incentive compatible regression learning.
\newblock In \emph{Proceedings of the nineteenth annual ACM-SIAM symposium on
  Discrete algorithms}, pages 884--893. Society for Industrial and Applied
  Mathematics, 2008.

\bibitem[Ghosh and Roth(2011)]{GR11}
Arpita Ghosh and Aaron Roth.
\newblock Selling privacy at auction.
\newblock In \emph{Proceedings of the 12th ACM Conference on Electronic
  Commerce (EC-11)}, 2011.

\bibitem[Ghosh et~al.(2014)Ghosh, Ligett, Roth, and
  Schoenebeck]{ghosh2014buying}
Arpita Ghosh, Katrina Ligett, Aaron Roth, and Grant Schoenebeck.
\newblock Buying private data without verification.
\newblock In \emph{Proceedings of the fifteenth ACM conference on Economics and
  computation}, pages 931--948. ACM, 2014.

\bibitem[Hanneke(2009)]{hanneke2009theoretical}
Steve Hanneke.
\newblock \emph{Theoretical foundations of active learning}.
\newblock ProQuest, 2009.

\bibitem[Horel et~al.(2014)Horel, Ioannidis, and Muthukrishnan]{HIM14}
Thibaut Horel, Stratis Ioannidis, and Muthu Muthukrishnan.
\newblock Budget feasible mechanisms for experimental design.
\newblock In \emph{Latin American Theoretical Informatics (LATIN-14)}, 2014.

\bibitem[Ligett and Roth(2012)]{LR12}
Katrina Ligett and Aaron Roth.
\newblock Take it or leave it: Running a survey when privacy comes at a cost.
\newblock In \emph{The 8th Workshop on Internet and Network Economics
  (WINE-12)}, 2012.

\bibitem[Meir et~al.(2012)Meir, Procaccia, and Rosenschein]{meir2012algorithms}
Reshef Meir, Ariel~D Procaccia, and Jeffrey~S Rosenschein.
\newblock Algorithms for strategyproof classification.
\newblock \emph{Artificial Intelligence}, 186:\penalty0 123--156, 2012.

\bibitem[Roth and Schoenebeck(2012)]{RothSchoenebeck12}
Aaron Roth and Grant Schoenebeck.
\newblock Conducting truthful surveys, cheaply.
\newblock In \emph{13th Conference on Electronic Commerce (EC-12)}, 2012.

\bibitem[Settles(2011)]{settles2011theories}
Burr Settles.
\newblock From theories to queries: Active learning in practice.
\newblock \emph{Active Learning and Experimental Design W}, pages 1--18, 2011.

\bibitem[Shalev-Shwartz(2012)]{Shai12}
Shai Shalev-Shwartz.
\newblock Online learning and online convex optimization.
\newblock \emph{Foundations and Trends in Machine Learning}, 2012.

\bibitem[Vapnik(2000)]{vapnik2000nature}
Vladimir Vapnik.
\newblock \emph{The nature of statistical learning theory}.
\newblock Springer Science \& Business Media, 2000.

\bibitem[Zinkevich(2003)]{zinkevich2003online}
Martin Zinkevich.
\newblock Online convex programming and generalized infinitesimal gradient
  ascent.
\newblock 2003.

\end{thebibliography}
